\documentclass[draft]{article}
\usepackage{amsfonts,amsmath,amssymb,color,graphicx,latexsym,multicol,times}
\newtheorem{example}{Example}
\newtheorem{claim}{Claim}
\newtheorem{lemma}{Lemma}
\newtheorem{proposition}{Proposition}
\newenvironment{proof}{{\bf Proof.}\ }{\hspace*{\fill}$\dashv$}
\title{Representation theorems for extended contact algebras based on equivalence relations}
\author{Philippe Balbiani$^{1}$ and Tatyana Ivanova$^{2}$}
\date{$^{1}$Institut de recherche en informatique de Toulouse
\\
CNRS --- Toulouse University
\\
$^{2}$Institute of Mathematics and Informatics
\\
Bulgarian Academy of Sciences}
\begin{document}
\maketitle
\begin{abstract}
The aim of this paper is to give new representation theorems for extended contact algebras.
These representation theorems are based on equivalence relations.
\end{abstract}
\def\R{\mbox{R\hspace{-.55em}R}}
\def\Q{\mbox{Q\hspace{-.55em}Q}}
\def\Z{\mbox{Z\hspace{-.55em}Z}}
\def\K{\mbox{I\hspace{-.15em}K}}
\def\N{\mathbb{N}}
\section{Introduction}\label{section:Introduction}
Starting with the belief that the spatial entities like points and lines usually considered in Euclidean geometry are too abstract, de Laguna~\cite{deLaguna} and Whitehead~\cite{Whitehead} have put forward other primitive entities like solids, or regions.
Between these entities, they have considered relations of ``connection'' (a ternary relation for de Laguna and a binary relation for Whitehead).
They have also axiomatically defined sets of properties that these relations should possess in order to provide an adequate analog of the reality we perceive about the connection relation between solids, or regions.
The ideas about space of de Laguna and Whitehead have enjoyed an uncommon destiny.
They constitute the basis of multifarious pointless theories of space since the days of Tarski's geometry of solids.
We can cite Grzegorczyk's theory of the binary relations of ``part-of'' and ``separation''~\cite{Grzegorczyk:part:of:separation}, de Vries' compingent algebras~\cite{de:Vries:compingent:algebras} based on a binary relation that would be called today ``non-tangential proper part'', etc.
See~\cite{Gerla:1995} for details.
\\
\\
The reason for the success of the axiomatic method within the context of the region-based theories of space certainly lies in the fact that our perception of space inevitably leads us to think out the relative positions of the objects that occupy space in terms of ``part-of'' and ``separation'', or in terms of ``part-of'' and ``connection''.
After the contributions of Clarke~\cite{Clarke:1981,Clark:1985}, several region-based theories of space have been developed in artificial intelligence and computer science within the context of the so-called area of qualitative spatial reasoning~\cite{Cohn:Renz:2008,Li:Ying:2003,Randell:Cui:Cohn:1992,Renz:2002,Renz:Nebel:1999}.
In these theories, one generally assumes that solids, or regions are regular closed subsets in, for example, the real plane together with its ordinary topology and one generally studies pointless theories of space based --- together with some other relations like ``partial overlap'', ``tangential proper part'', etc --- on the binary relation of ``contact'' which holds between two regular closed subsets when they have common points.
\\
\\
There are mainly two kinds of results: representability in concrete geometrical structures like topological spaces of abstract algebraic structures like contact algebras~\cite{Dimov:Vakarelov:2006:I,Dimov:Vakarelov:2006:II,Duentsch:Vakarelov:2007,Duentsch:Winter:2005}; computability of the satisfiability problem with respect to a quantifier-free first-order language to be interpreted in such contact algebras~\cite{Kontchakov:Nenov:Pratt:Hartmann:Zakharyaschev:2013,Kontchakov:Pratt:Hartmann:Wolter:Zakharyaschev:2010,Kontchakov:Pratt:Hartmann:Zakharyaschev:2010,Kontchakov:Pratt:Hartmann:Zakharyaschev:2014}.
In this context, the unary relation of ``internal connectedness'' has been considered which holds for those regular closed subsets which cannot be represented as the union of two disjoint nonempty open sets.
See the above-mentioned references for details.
As remarked by Ivanova~\cite[Chapter~$2$]{Ivanova:2016} in her doctoral thesis, this unary relation cannot be elementarily defined in terms of the binary relation of ``contact'' within the class of all topological spaces.
This has led her to introduce the ternary relation of ``covering'' which holds between three regular closed subsets when the points common to the first two subsets belong to the third subset.
See also Vakarelov~\cite{Vakarelov:2018} for a multi-ary version of this relation.
\\
\\
By using techniques based on the theory of filters and ideals, Ivanova~\cite[Chapter~$2$]{Ivanova:2016} has proved the representability in ordinary topological spaces of the extended contact algebras that she has defined.
As suggested by the ideas of Galton~\cite{Galton:1999,Galton:2000} and Vakarelov~\cite{Vakarelov:1997}, representability in concrete relational structures like Kripke frames of abstract algebraic structures like contact algebras might be obtained too.
The aim of this paper is to give new representability this time in concrete relational structures for extended contact algebras.
In Section~\ref{section:Contact:and:extended:contact:relations}, we introduce contact and extended contact relations between regions in topological spaces.
Section~\ref{section:Contact:and:extended:contact:algebras} defines contact and extended contact algebras and discusses about their topological and relational representations.
In Sections~\ref{subsection:Equivalence:frames:of:type:1} and~\ref{subsection:Equivalence:frames:of:type:2}, two other kinds of extended contact algebra based on equivalence relations are introduced and the representability of extended contact algebras in them is proved.
\section{Contact and extended contact relations}\label{section:Contact:and:extended:contact:relations}
In this section, we introduce the contact and extended contact relations between regular closed subsets of topological spaces.
\subsection{Topological spaces}
A {\it topological space} is a structure of the form $(X,\tau)$ where $X$ is a nonempty set and $\tau$ is a {\it topology} on $X$, i.e. a set of subsets of $X$ such that the following conditions hold:
\begin{itemize}
\item $\emptyset$ is in $\tau$,
\item $X$ is in $\tau$,
\item if $\{A_{i}:\ i\in I\}$ is a finite subset of $\tau$ then $\bigcap\{A_{i}:\ i\in I\}$ is in $\tau$,
\item if $\{A_{i}:\ i\in I\}$ is a subset of $\tau$ then $\bigcup\{A_{i}:\ i\in I\}$ is in $\tau$.
\end{itemize}
The subsets of $X$ in $\tau$ are called {\it open sets} whereas their complements are called {\it closed sets.}
Let $(X,\tau)$ be a topological space.
For all subsets $A$ of $X$, the {\it interior} of $A$ (denoted $\operatorname{Int}_{\tau}(A)$) is the union of the open subsets $B$ of $X$ such that $B\subseteq A$.
It is the greatest open set contained in $A$.
For all subsets $A$ of $X$, the {\it closure} of $A$ (denoted $\operatorname{Cl}_{\tau}(A)$) is the intersection of the closed subsets $B$ of $X$ such that $A\subseteq B$.
It is the least closed set containing $A$.
A subset $A$ of $X$ is {\it regular closed} iff $\operatorname{Cl}_{\tau}(\operatorname{Int}_{\tau}(A))=A$.
Regular closed subsets of $X$ will also be called {\it regions.}
It is well-known that the set $RC(X,\tau)$ of all regular closed subsets of $X$ forms a Boolean algebra $(RC(X,\tau),0_{X},\star_{X},\cup_{X})$ where for all $A,B\in RC(X,\tau)$,
\begin{itemize}
\item $0_{X}=\emptyset$,
\item $A^{\star_{X}}=\operatorname{Cl}_{\tau}(X\setminus A)$,
\item $A\cup_{X}B=A\cup B$.
\end{itemize}
In this Boolean algebra, of course, we have for all $A,B\in RC(X,\tau)$, $1_{X}=0_{X}^{\star_{X}}$ and $A\cap_{X}B=(A^{\star_{X}}\cup_{X}B^{\star_{X}})^{\star_{X}}$, i.e. $1_{X}=X$ and $A\cap_{X}B=\operatorname{Cl}_{\tau}(\operatorname{Int}_{\tau}(A\cap B))$.
\subsection{Standard contact algebra of regular closed sets}
Given a topological space $(X,\tau)$, since regions are regular closed subsets of $X$, therefore two regions are {\it in contact} iff they have a nonempty intersection.
For this reason, we define the binary relation $C_{X}$ on $RC(X,\tau)$ by
\begin{itemize}
\item $C_{X}(A,B)$ iff $A$ and $B$ have a nonempty intersection.
\end{itemize}
The relation $C_{X}$ is called {\it contact relation} on $RC(X,\tau)$ and we read $C_{X}(A,B)$ as follows: ``$A$ and $B$ are in contact''.
As for the structure $(RC(X,\tau),0_{X},\star_{X},\cup_{X},C_{X})$ based on the set $RC(X,\tau)$ of all regular closed subsets of $X$, it is called the {\it standard contact algebra of regular closed sets.}
It has been studied at great length within the context of first-order mereotopologies~\cite{Pratt:Hartmann:2007} and region-based theories of space~\cite{Balbiani:Tinchev:Vakarelov:2007,Vakarelov:2007}.
Another structure, this time based on the set $RO(X,\tau)$ of all regular open subsets of $X$, i.e. those subsets $A$ of $X$ such that $\operatorname{Int}_{\tau}(\operatorname{Cl}_{\tau}(A))=A$, can be defined as well.
It is the structure $(RO(X,\tau),0_{X},\star_{X},\cup_{X},C_{X})$ called the {\it standard contact algebra of regular open sets} and where for all $A,B\in RO(X,\tau)$,
\begin{itemize}
\item $0_{X}=\emptyset$,
\item $A^{\star_{X}}=\operatorname{Int}_{\tau}(X\setminus A)$,
\item $A\cup_{X}B=\operatorname{Int}_{\tau}(\operatorname{Cl}_{\tau}(A\cup B))$,
\item $C_{X}(A,B)$ iff $\operatorname{Cl}_{\tau}(A)$ and $\operatorname{Cl}_{\tau}(B)$ have a nonempty intersection.
\end{itemize}
At the Boolean level, of course, we have for all $A,B\in RO(X,\tau)$, $1_{X}=0_{X}^{\star_{X}}$ and $A\cap_{X}B=(A^{\star_{X}}\cup_{X}B^{\star_{X}})^{\star_{X}}$, i.e. $1_{X}=X$ and $A\cap_{X}B=A\cap B$.
Since an arbitrary standard contact algebra of regular open sets is isomorphic to the corresponding standard contact algebra of regular closed sets, therefore in this paper, we will only interest with the standard contact algebras of regular closed sets.
In order to give a flavor of the properties of the contact relation, let us remark that for all $A,B,D,E\in RC(X,\tau)$,
\begin{itemize}
\item if $C_{X}(A,B)$, $A\subseteq D$ and $B\subseteq E$ then $C_{X}(D,E)$,
\item if $C_{X}(A\cup_{X}D,B\cup_{X}E)$ then $C_{X}(A,B)$, or $C_{X}(A,E)$, or $C_{X}(D,B)$, or $C_{X}(D,E)$,
\item if $C_{X}(A,B)$ then $A\not=0_{X}$ and $B\not=0_{X}$,
\item if $A\not=0_{X}$ then $C_{X}(A,A)$,
\item if $C_{X}(A,B)$ then $C_{X}(B,A)$.
\end{itemize}
These conditions, or equivalent ones, have given rise in~\cite{Dimov:Vakarelov:2006:I,Dimov:Vakarelov:2006:II} to the algebras of regions known as contact algebras.
See also~\cite{Duentsch:Vakarelov:2007,Duentsch:Winter:2005}.
Representation theorems establishing a correspondence between region-based models such as contact algebras and point-based models such as topological spaces have been obtained.
See Subsection~\ref{subsection:topological:representation:contact:algebras}.
\subsection{Internal connectedness}
Within the context of topological logics over Euclidean spaces~\cite{Kontchakov:Nenov:Pratt:Hartmann:Zakharyaschev:2013,Kontchakov:Pratt:Hartmann:Wolter:Zakharyaschev:2010,Kontchakov:Pratt:Hartmann:Zakharyaschev:2010,Kontchakov:Pratt:Hartmann:Zakharyaschev:2014,Tinchev:Vakarelov:2010}, the {\it relation of internal connectedness} has been considered too.
Given a topological space $(X,\tau)$, the relation of internal connectedness is the unary relation $c^{\circ}_{X}$ on $RC(X,\tau)$ defined by
\begin{itemize}
\item $c^{\circ}_{X}(A)$ iff $\operatorname{Int}_{\tau}(A)$ is connected, i.e. $\operatorname{Int}_{\tau}(A)$ cannot be represented as the union of two disjoint nonempty open sets.
\end{itemize}
We read $c^{\circ}_{X}(A)$ as follows: ``$A$ is internally connected''.
Immediately, the question arises as to whether the relation of internal connectedness can be elementarily defined in terms of the contact relation within the class of all topological spaces.
This question has been negatively answered.
\begin{example}\label{example:about:internal:connectedness:and:contact}
Let $X=\{1,2,3,4,5,6,7\}$ and $\tau$ be the least topology on $X$ containing $\{1,2,3\}$, $\{2,5,7\}$ and $\{3,6,7\}$.
Let $X^{\prime}=\{2,3,4,5,6,7\}$ and $\tau^{\prime}$ be the least topology on $X^{\prime}$ containing $\{2,3\}$, $\{2,5,7\}$ and $\{3,6,7\}$.
Obviously, for all subsets $A$ of $X$, $A\in RC(X,\tau)$ iff $A\setminus\{1\}\in RC(X^{\prime},\tau^{\prime})$.
Moreover, for all $A,B\in RC(X,\tau)$, $C_{X}(A,B)$ iff $C_{X^{\prime}}(A\setminus\{1\},B\setminus\{1\})$.
Hence, the function $f$ associating to each regular closed subset of $X$ a regular closed subset of $X^{\prime}$ defined by $f(A)=A\setminus\{1\}$ is an isomorphism from $(RC(X,\tau),0_{X},\star_{X},\cup_{X},C_{X})$ to $(RC(X^{\prime},\tau^{\prime}),0_{X^{\prime}},\star_{X^{\prime}},\cup_{X^{\prime}},C_{X^{\prime}})$.
Since the regular closed subset $\{1,2,3,4,5,6\}$ of $X$ is internally connected and the regular closed subset $\{2,3,4,5,6\}$ of $X^{\prime}$ is not internally connected, therefore the relation of internal connectedness cannot be elementarily defined in terms of the contact relation within the class of all topological spaces.
See~\cite[Chapter~$2$]{Ivanova:2016} for details.
\end{example}
\subsection{Covering}
Example~\ref{example:about:internal:connectedness:and:contact} has led, given a topological space $(X,\tau)$, to the {\it relation of covering} on $RC(X,\tau)$.
It is the ternary relation $\vdash_{X}$ on $RC(X,\tau)$ defined in~\cite[Chapter~$2$]{Ivanova:2016} by
\begin{itemize}
\item $(A,B)\vdash_{X}D$ iff the intersection of $A$ and $B$ is included in $D$.
\end{itemize}
We read $(A,B)\vdash_{X}D$ as follows: ``$A$ and $B$ are covered by $D$''.
The relation $\vdash_{X}$ is also called {\it extended contact relation} on $RC(X,\tau)$.
Obviously, the contact relation can be elementarily defined in terms of the relation of covering within the class of all topological spaces, seeing that for all $A,B\in RC(X,\tau)$,
\begin{itemize}
\item $C_{X}(A,B)$ iff $(A,B)\not\vdash_{X}\emptyset$.
\end{itemize}
More interestingly, it turns out that the relation of internal connectedness can be as well elementarily defined in terms of the relation of covering within the class of all topological spaces, seeing that for all $A\in RC(X,\tau)$,
\begin{itemize}
\item $c^{\circ}_{X}(A)$ iff for all $B,D\in RC(X,\tau)$, $B,D\not=\emptyset$, if $A=B\cup_{X}D$ then $(B,D)\not\vdash_{X}A^{\star_{X}}$.
\end{itemize}
Since the relation of internal connectedness cannot be elementarily defined in terms of the contact relation within the class of all topological spaces, therefore the relation of covering cannot be elementarily defined in terms of the contact relation within the class of all topological spaces.
The question as to whether the contact relation can be elementarily defined in terms of the relation of internal connectedness within the class of all topological spaces is still open.
In order to give a flavor of the properties of the relation of covering, let us remark that for all $A,B,D,E,F\in RC(X,\tau)$,
\begin{itemize}
\item if $(A,B)\vdash_{X}F$ then $(A\cup_{X}D,B\cup_{X}E)\vdash_{X}D\cup_{X}E\cup_{X}F$,
\item if $(A,B)\vdash_{X}D$, $(A,B)\vdash_{X}E$ and $(D,E)\vdash_{X}F$ then $(A,B)\vdash_{X}F$,
\item if $A\subseteq F$, or $B\subseteq F$ then $(A,B)\vdash_{X}F$,
\item if $(A,B)\vdash_{X}F$ then $A\cap_{X}B\subseteq F$,
\item if $(A,B)\vdash_{X}F$ then $(B,A)\vdash_{X}F$.
\end{itemize}
These conditions, or equivalent ones, have given rise in~\cite[Chapter~$2$]{Ivanova:2016} to the algebras of regions known as extended contact algebras.
Representation theorems establishing a correspondence between region-based models such as extended contact algebras and point-based models such as topological spaces have been obtained.
See Subsection~\ref{subsection:topological:representation:extended:contact:algebras}.
\section{Contact and extended contact algebras}\label{section:Contact:and:extended:contact:algebras}
In this section, we introduce contact and extended contact algebras and we discuss about their topological and relational representations.
\subsection{Contact algebras}
After~\cite{Dimov:Vakarelov:2006:I,Dimov:Vakarelov:2006:II}, a {\it contact algebra} is a structure of the form $({\mathcal R},0_{{\mathcal R}},\star_{{\mathcal R}},\cup_{{\mathcal R}},C_{{\mathcal R}})$ where $({\mathcal R},0_{{\mathcal R}},\star_{{\mathcal R}},\cup_{{\mathcal R}})$ is a non-degenerate Boolean algebra and $C_{{\mathcal R}}$ is a binary relation on ${\mathcal R}$ such that for all $a,b,d,e\in{\mathcal R}$,
\begin{description}
\item[$(CA_{1})$] if $C_{{\mathcal R}}(a,b)$, $a\leq_{{\mathcal R}}d$ and $b\leq_{{\mathcal R}}e$ then $C_{{\mathcal R}}(d,e)$,
\item[$(CA_{2})$] if $C_{{\mathcal R}}(a\cup_{{\mathcal R}}d,b\cup_{{\mathcal R}}e)$ then $C_{{\mathcal R}}(a,b)$, or $C_{{\mathcal R}}(a,e)$, or $C_{{\mathcal R}}(d,b)$, or $C_{{\mathcal R}}(d,e)$,
\item[$(CA_{3})$] if $C_{{\mathcal R}}(a,b)$ then $a\not=0_{{\mathcal R}}$ and $b\not=0_{{\mathcal R}}$,
\item[$(CA_{4})$] if $a\not=0_{{\mathcal R}}$ then $C_{{\mathcal R}}(a,a)$,
\item[$(CA_{5})$] if $C_{{\mathcal R}}(a,b)$ then $C_{{\mathcal R}}(b,a)$.
\end{description}
At the Boolean level, of course, we have for all $a,b\in{\mathcal R}$, $1_{{\mathcal R}}=0_{{\mathcal R}}^{\star_{{\mathcal R}}}$ and $a\cap_{{\mathcal R}}b=(a^{\star_{{\mathcal R}}}\cup_{{\mathcal R}}b^{\star_{{\mathcal R}}})^{\star_{{\mathcal R}}}$.
The elements of ${\mathcal R}$ are called {\it regions.}
\subsection{Topological representation of contact algebras}\label{subsection:topological:representation:contact:algebras}
We have seen, in Section~\ref{section:Contact:and:extended:contact:relations}, that for all topological spaces $(X,\tau)$, the structure $(RC(X,
$\linebreak$
\tau),0_{X},\star_{X},\cup_{X},C_{X})$ based on the set $RC(X,\tau)$ of all regular closed subsets of $X$ is a contact algebra.
With the following proposition~\cite{Dimov:Vakarelov:2006:I,Dimov:Vakarelov:2006:II,Duentsch:Winter:2005}, one can say that standard contact algebras of regular closed sets are typical examples of contact algebras.
\begin{proposition}\label{proposition:topological:representations:of:contact:algebras}
Let $({\mathcal R},0_{{\mathcal R}},\star_{{\mathcal R}},\cup_{{\mathcal R}},C_{{\mathcal R}})$ be a contact algebra.
There exists a topological space $(X,\tau)$ and an embedding of $({\mathcal R},0_{{\mathcal R}},\star_{{\mathcal R}},\cup_{{\mathcal R}},C_{{\mathcal R}})$ in $(RC(X,\tau),0_{X},\star_{X},
$\linebreak$
\cup_{X},C_{X})$.
Moreover, if ${\mathcal R}$ is finite then $X$ is finite and $h$ is surjective.
\end{proposition}
\subsection{Relational representation of contact algebras}
Another kind of contact algebra has been independently considered by Galton~\cite{Galton:1999,Galton:2000} and Vakarelov~\cite{Vakarelov:1997}.
A {\it frame} is a structure of the form $(W,R)$ where $W$ is a nonempty set and $R$ is a binary relation on $W$.
Given a frame $(W,R)$, let $C_{W}$ be the binary relation on $W$'s powerset defined by
\begin{itemize}
\item $C_{W}(A,B)$ iff there exists $s\in A$ and there exists $t\in B$ such that $R(s,t)$.
\end{itemize}
Assuming that $R$ is reflexive and symmetric, the reader may easily verify that the structure $({\mathcal P}(W),0_{W},\star_{W},\cup_{W},C_{W})$, where $0_{W}$ is the empty set, $\star_{W}$ is the complement operation with respect to $W$ and $\cup_{W}$ is the union operation, is a contact algebra.
At the Boolean level, of course, we have for all $A,B\in{\mathcal P}(W)$, $1_{W}=0_{W}^{\star_{W}}$ and $A\cap_{W}B=(A^{\star_{W}}\cup_{W}B^{\star_{W}})^{\star_{W}}$, i.e. $1_{W}=W$ and $A\cap_{W}B=A\cap B$.
With the following proposition~\cite{Duentsch:Vakarelov:2007}, one can say that these contact algebras are typical examples of contact algebras as well.
\begin{proposition}\label{proposition:relational:representations:of:contact:algebras}
Let $({\mathcal R},0_{{\mathcal R}},\star_{{\mathcal R}},\cup_{{\mathcal R}},C_{{\mathcal R}})$ be a contact algebra.
There exists a frame $(W,R)$ and an embedding of $({\mathcal R},0_{{\mathcal R}},\star_{{\mathcal R}},\cup_{{\mathcal R}},C_{{\mathcal R}})$ in $({\mathcal P}(W),0_{W},\star_{W},\cup_{W},C_{W})$.
\end{proposition}
\subsection{Extended contact algebras}
After~\cite[Chapter~$2$]{Ivanova:2016}, an {\it extended contact algebra} is a structure of the form $({\mathcal R},0_{{\mathcal R}},
$\linebreak$
\star_{{\mathcal R}},\cup_{{\mathcal R}},\vdash_{{\mathcal R}})$ where $({\mathcal R},0_{{\mathcal R}},\star_{{\mathcal R}},\cup_{{\mathcal R}})$ is a non-degenerate Boolean algebra and $\vdash_{{\mathcal R}}$ is a ternary relation on ${\mathcal R}$ such that for all $a,b,d,e,f\in{\mathcal R}$,
\begin{description}
\item[$(ECA_{1})$] if $(a,b)\vdash_{{\mathcal R}}f$ then $(a\cup_{{\mathcal R}}d,b\cup_{{\mathcal R}}e)\vdash_{{\mathcal R}}d\cup_{{\mathcal R}}e\cup_{{\mathcal R}}f$,
\item[$(ECA_{2})$] if $(a,b)\vdash_{{\mathcal R}}d$, $(a,b)\vdash_{{\mathcal R}}e$ and $(d,e)\vdash_{{\mathcal R}}f$ then $(a,b)\vdash_{{\mathcal R}}f$,
\item[$(ECA_{3})$] if $a\leq_{{\mathcal R}}f$, or $b\leq_{{\mathcal R}}f$ then $(a,b)\vdash_{{\mathcal R}}f$,
\item[$(ECA_{4})$] if $(a,b)\vdash_{{\mathcal R}}f$ then $a\cap_{{\mathcal R}}b\leq_{{\mathcal R}}f$,
\item[$(ECA_{5})$] if $(a,b)\vdash_{{\mathcal R}}f$ then $(b,a)\vdash_{{\mathcal R}}f$.
\end{description}
At the Boolean level, of course, we have for all $a,b\in{\mathcal R}$, $1_{{\mathcal R}}=0_{{\mathcal R}}^{\star_{{\mathcal R}}}$ and $a\cap_{{\mathcal R}}b=(a^{\star_{{\mathcal R}}}\cup_{{\mathcal R}}b^{\star_{{\mathcal R}}})^{\star_{{\mathcal R}}}$.
Again, the elements of ${\mathcal R}$ are called {\it regions.}
It is remarkable that the binary relation $C_{{\mathcal R}}$ on ${\mathcal R}$ defined as follows satisfies the conditions $(CA_{1})$--$(CA_{5})$:
\begin{itemize}
\item $C_{{\mathcal R}}(a,b)$ iff $(a,b)\not\vdash_{{\mathcal R}}0_{{\mathcal R}}$.
\end{itemize}
This leads us to associate to each $g\in{\mathcal R}$, the binary relation $RC_{{\mathcal R}}(g)$ on ${\mathcal R}$ of {\it relative contact} defined as follows:
\begin{itemize}
\item $RC_{{\mathcal R}}(g)(a,b)$ iff $(a,b)\not\vdash_{{\mathcal R}}g$.
\end{itemize}
This definition of the binary relation $RC_{{\mathcal R}}(g)$ on ${\mathcal R}$ associated to each $g\in{\mathcal R}$ has interesting consequences.
\begin{proposition}\label{proposition:consequences:definition:rc}
Let $({\mathcal R},0_{{\mathcal R}},\star_{{\mathcal R}},\cup_{{\mathcal R}},\vdash_{{\mathcal R}})$ be an extended contact algebra.
For all $a,b,
$\linebreak$
d,e,g\in{\mathcal R}$, the following conditions hold:
\begin{enumerate}
\item If $RC_{{\mathcal R}}(g)(a,b)$, $a\leq_{{\mathcal R}}d$ and $b\leq_{{\mathcal R}}e$ then $RC_{{\mathcal R}}(g)(d,e)$.
\item If $RC_{{\mathcal R}}(g)(a\cup_{{\mathcal R}}d,b\cup_{{\mathcal R}}e)$ then $RC_{{\mathcal R}}(g)(a,b)$, or $RC_{{\mathcal R}}(g)(a,e)$, or $RC_{{\mathcal R}}(g)(d,
$\linebreak$
b)$, or $RC_{{\mathcal R}}(g)(d,e)$.
\item If $RC_{{\mathcal R}}(g)(a,b)$ then $a\not\leq_{{\mathcal R}}g$ and $b\not\leq_{{\mathcal R}}g$.
\item If $a\not\leq_{{\mathcal R}}g$ then $RC_{{\mathcal R}}(g)(a,a)$.
\item If $RC_{{\mathcal R}}(g)(a,b)$ then $RC_{{\mathcal R}}(g)(b,a)$.
\end{enumerate}
\end{proposition}
\begin{proof}
$(1)$
Suppose $RC_{{\mathcal R}}(g)(a,b)$, $a\leq_{{\mathcal R}}d$ and $b\leq_{{\mathcal R}}e$.
Hence, $(a,b)\not\vdash_{{\mathcal R}}g$.
By items~$(5)$ and~$(6)$ of Proposition~\ref{proposition:consequences:definition:eca}, $(d,e)\not\vdash_{{\mathcal R}}g$.
Thus, $RC_{{\mathcal R}}(g)(d,e)$.
\\
\\
$(2)$
Suppose $RC_{{\mathcal R}}(g)(a\cup_{{\mathcal R}}d,b\cup_{{\mathcal R}}e)$.
Hence, $(a\cup_{{\mathcal R}}d,b\cup_{{\mathcal R}}e)\not\vdash_{{\mathcal R}}g$.
By item~$(3)$ of Proposition~\ref{proposition:consequences:definition:eca}, $(a,b\cup_{{\mathcal R}}e)\not\vdash_{{\mathcal R}}g$, or $(d,b\cup_{{\mathcal R}}e)\not\vdash_{{\mathcal R}}g$.
In the former case, by item~$(4)$ of Proposition~\ref{proposition:consequences:definition:eca}, $(a,b)\not\vdash_{{\mathcal R}}g$, or $(a,e)\not\vdash_{{\mathcal R}}g$.
Thus, $RC_{{\mathcal R}}(g)(a,b)$, or $RC_{{\mathcal R}}(g)(a,e)$.
In the latter case, by item~$(4)$ of Proposition~\ref{proposition:consequences:definition:eca}, $(d,b)\not\vdash_{{\mathcal R}}g$, or $(d,e)\not\vdash_{{\mathcal R}}g$.
Consequently, $RC_{{\mathcal R}}(g)(d,b)$, or $RC_{{\mathcal R}}(g)(d,e)$.
In both cases, $RC_{{\mathcal R}}(g)(a,b)$, or $RC_{{\mathcal R}}(g)(a,e)$, or $RC_{{\mathcal R}}(g)(d,b)$, or $RC_{{\mathcal R}}(g)(d,e)$.
\\
\\
$(3)$
Suppose $RC_{{\mathcal R}}(g)(a,b)$.
Hence, $(a,b)\not\vdash_{{\mathcal R}}g$.
By~$(ECA_{3})$, $a\not\leq_{{\mathcal R}}g$ and $b\not\leq_{{\mathcal R}}g$.
\\
\\
$(4)$
Suppose $a\not\leq_{{\mathcal R}}g$.
By~$(ECA_{4})$, $(a,a)\not\vdash_{{\mathcal R}}g$.
Hence, $RC_{{\mathcal R}}(g)(a,a)$.
\\
\\
$(5)$
Suppose $RC_{{\mathcal R}}(g)(a,b)$.
Hence, $(a,b)\not\vdash_{{\mathcal R}}g$.
By~$(ECA_{5})$, $(b,a)\not\vdash_{{\mathcal R}}g$.
Thus, $RC_{{\mathcal R}}(g)(b,a)$.
\end{proof}
\subsection{Topological representation of extended contact algebras}\label{subsection:topological:representation:extended:contact:algebras}
We have seen, in Section~\ref{section:Contact:and:extended:contact:relations}, that for all topological spaces $(X,\tau)$, the structure $(RC(X,
$\linebreak$
\tau),0_{X},\star_{X},\cup_{X},\vdash_{X})$ based on the set $RC(X,\tau)$ of all regular closed subsets of $X$ is an extended contact algebra.
With the following proposition~\cite[Chapter~$2$]{Ivanova:2016}, one can say that standard extended contact algebras of regular closed sets are typical examples of extended contact algebras.
See also~\cite{Vakarelov:2018}.
\begin{proposition}\label{proposition:topological:representations:of:extended:contact:algebras}
Let $({\mathcal R},0_{{\mathcal R}},\star_{{\mathcal R}},\cup_{{\mathcal R}},\vdash_{{\mathcal R}})$ be an extended contact algebra.
There exists a topological space $(X,\tau)$ and an embedding of $({\mathcal R},0_{{\mathcal R}},\star_{{\mathcal R}},\cup_{{\mathcal R}},\vdash_{{\mathcal R}})$ in $(RC(X,\tau),
$\linebreak$
0_{X},\star_{X},\cup_{X},\vdash_{X})$.
Moreover, if ${\mathcal R}$ is finite then $X$ is finite and $h$ is surjective.
\end{proposition}
\subsection{Relational representation of extended contact algebras}
Another kind of extended contact algebra based on parametrized frames can be considered.
A {\it weak extended contact algebra} is a structure of the form $({\mathcal R},0_{{\mathcal R}},\star_{{\mathcal R}},\cup_{{\mathcal R}},\vdash_{{\mathcal R}})$ where $({\mathcal R},0_{{\mathcal R}},\star_{{\mathcal R}},\cup_{{\mathcal R}})$ is a non-degenerate Boolean algebra and $\vdash_{{\mathcal R}}$ is a ternary relation on ${\mathcal R}$ such that for all $a,b,d,e,f\in{\mathcal R}$,
\begin{description}
\item[$(WECA_{1})$] if $a\leq_{{\mathcal R}}d$, $b\leq_{{\mathcal R}}e$ and $(d,e)\vdash_{{\mathcal R}}f$ then $(a,b)\vdash_{{\mathcal R}}f$,
\item[$(WECA_{2})$] if $a=0_{{\mathcal R}}$, or $b=0_{{\mathcal R}}$ then $(a,b)\vdash_{{\mathcal R}}f$,
\item[$(WECA_{3})$] if $(a,b)\vdash_{{\mathcal R}}f$ and $(d,e)\vdash_{{\mathcal R}}f$ then $(a\cap_{{\mathcal R}}d,b\cup_{{\mathcal R}}e)\vdash_{{\mathcal R}}f$ and $(a\cup_{{\mathcal R}}d,b\cap_{{\mathcal R}}e)\vdash_{{\mathcal R}}f$,
\item[$(WECA_{4})$] if $(a,b)\vdash_{{\mathcal R}}d$ and $d\leq_{{\mathcal R}}f$ then $(a,b)\vdash_{{\mathcal R}}f$.
\end{description}
At the Boolean level, of course, we have for all $a,b\in{\mathcal R}$, $1_{{\mathcal R}}=0_{{\mathcal R}}^{\star_{{\mathcal R}}}$ and $a\cap_{{\mathcal R}}b=(a^{\star_{{\mathcal R}}}\cup_{{\mathcal R}}b^{\star_{{\mathcal R}}})^{\star_{{\mathcal R}}}$.
Again, the elements of ${\mathcal R}$ are called {\it regions.}
Obviously, every extended contact algebra is also a weak extended contact algebra.
What is more, the conditions $(WECA_{1})$--$(WECA_{4})$ have interesting consequences.
\begin{proposition}\label{proposition:consequences:definition:eca}
Let $({\mathcal R},0_{{\mathcal R}},\star_{{\mathcal R}},\cup_{{\mathcal R}},\vdash_{{\mathcal R}})$ be a weak extended contact algebra.
For all $a,b,d,e,f\in{\mathcal R}$, the following conditions hold:
\begin{enumerate}
\item $(0_{{\mathcal R}},1_{{\mathcal R}})\vdash_{{\mathcal R}}a$.
\item $(1_{{\mathcal R}},0_{{\mathcal R}})\vdash_{{\mathcal R}}a$.
\item if $(a,d)\vdash_{{\mathcal R}}e$ and $(b,d)\vdash_{{\mathcal R}}e$ then $(a\cup_{{\mathcal R}}b,d)\vdash_{{\mathcal R}}e$.
\item if $(a,b)\vdash_{{\mathcal R}}e$ and $(a,d)\vdash_{{\mathcal R}}e$ then $(a,b\cup_{{\mathcal R}}d)\vdash_{{\mathcal R}}e$.
\item if $(a,b)\vdash_{{\mathcal R}}e$ and $d\leq_{{\mathcal R}}a$ then $(d,b)\vdash_{{\mathcal R}}e$.
\item if $(a,b)\vdash_{{\mathcal R}}e$ and $d\leq_{{\mathcal R}}b$ then $(a,d)\vdash_{{\mathcal R}}e$.
\end{enumerate}
\end{proposition}
\begin{proof}
$(1)$
By~$(WECA_{2})$, $(0_{{\mathcal R}},1_{{\mathcal R}})\vdash_{{\mathcal R}}a$.
\\
\\
$(2)$
Similar to $(1)$.
\\
\\
$(3)$
Suppose $(a,d)\vdash_{{\mathcal R}}e$ and $(b,d)\vdash_{{\mathcal R}}e$.
By~$(WECA_{3})$, $(a\cup_{{\mathcal R}}b,d)\vdash_{{\mathcal R}}e$.
\\
\\
$(4)$
Similar to $(3)$.
\\
\\
$(5)$
Suppose $(a,b)\vdash_{{\mathcal R}}e$ and $d\leq_{{\mathcal R}}a$.
By~$(WECA_{1})$, $(d,b)\vdash_{{\mathcal R}}e$.
\\
\\
$(6)$
Similar to $(5)$.
\end{proof}
\\
\\
A {\it parametrized frame} is a structure of the form $(W,R)$ where $W$ is a nonempty set and $R$ is a function associating to each subset of $W$ a binary relation on $W$.
Given a parametrized frame $(W,R)$, let $\vdash_{W}$ be the ternary relation on $W$'s powerset defined by
\begin{itemize}
\item $(A,B)\vdash_{W}D$ iff for all $s\in A$, for all $t\in B$ and for all $U\subseteq W$, if $D\subseteq U$ then not-$R(U)(s,t)$.
\end{itemize}
The reader may easily verify that the structure $({\mathcal P}(W),0_{W},\star_{W},\cup_{W},\vdash_{W})$, where, again, $0_{W}$ is the empty set, $\star_{W}$ is the complement operation with respect to $W$ and $\cup_{W}$ is the union operation, is a weak extended contact algebra.
With the following proposition, one can say that these weak extended contact algebras are typical examples of weak extended contact algebras as well.
\begin{proposition}\label{proposition:relational:representations:of:extended:contact:algebras}
Let $({\mathcal R},0_{{\mathcal R}},\star_{{\mathcal R}},\cup_{{\mathcal R}},\vdash_{{\mathcal R}})$ be a weak extended contact algebra.
There exists a parametrized frame $(W,R)$ and an embedding of $({\mathcal R},0_{{\mathcal R}},\star_{{\mathcal R}},\cup_{{\mathcal R}},\vdash_{{\mathcal R}})$ in $({\mathcal P}(W),0_{W},\star_{W},\cup_{W},\vdash_{W})$.
\end{proposition}
\begin{proof}
Let $(W,R)$ be the structure where
\begin{itemize}
\item $W$ is the set of all maximal filters in the Boolean algebra $({\mathcal R},0_{{\mathcal R}},\star_{{\mathcal R}},\cup_{{\mathcal R}})$,
\item $R$ is the function associating to each subset $U$ of $W$ the binary relation $R(U)$ on $W$ defined by $R(U)(s,t)$ iff for all $a,b,d\in{\mathcal R}$, if $a\in s$, $b\in t$ and $(a,b)\vdash_{{\mathcal R}}d$ then there exists $e\in{\mathcal R}$ such that $d\not\leq_{{\mathcal R}}e$ and for all $u\in U$, $e\in u$.
\end{itemize}
Obviously, $(W,R)$ is a parametrized frame.
Let $h$ be the function associating to each region in ${\mathcal R}$ a subset of $W$ defined by
\begin{itemize}
\item $h(a)$ is the set of all $s\in W$ such that $a\in s$.
\end{itemize}
In order to prove that $h$ is an embedding of $({\mathcal R},0_{{\mathcal R}},\star_{{\mathcal R}},\cup_{{\mathcal R}},\vdash_{{\mathcal R}})$ in $({\mathcal P}(W),0_{W},\star_{W},
$\linebreak$
\cup_{W},\vdash_{W})$, let us prove the following
\begin{lemma}\label{lemma:for:the:proposition:relational:representations:of:extended:contact:algebras}
The following conditions hold:
\begin{enumerate}
\item $h$ is injective.
\item $h(0_{{\mathcal R}})=0_{W}$.
\item For all regions $a$ in ${\mathcal R}$, $h(a^{\star_{{\mathcal R}}})=h(a)^{\star_{W}}$.
\item For all regions $a,b$ in ${\mathcal R}$, $h(a\cup_{{\mathcal R}}b)=h(a)\cup_{W}h(b)$.
\item For all regions $a,b,d$ in ${\mathcal R}$, $(a,b)\vdash_{{\mathcal R}}d$ iff $(h(a),h(b))\vdash_{W}h(d)$.
\end{enumerate}
\end{lemma}
\begin{proof}
$(1)$--$(4)$~The injectivity of $h$ and the fact that $h$ preserves the operations $0$, $\star$ and $\cup$ follow from classical results in the theory of filters and ideals~\cite{Givant:Halmos:2009}.
\\
\\
$(5)$~Let $a,b,d$ be regions in ${\mathcal R}$.
We demonstrate $(a,b)\vdash_{{\mathcal R}}d$ iff $(h(a),h(b))\vdash_{W}h(d)$.
Suppose $(a,b)\vdash_{{\mathcal R}}d$ not-iff $(h(a),h(b))\vdash_{W}h(d)$.
Hence, $(a,b)\vdash_{{\mathcal R}}d$ and $(h(a),h(b))\not\vdash_{W}h(d)$, or $(h(a),h(b))\vdash_{W}h(d)$ and $(a,b)\not\vdash_{{\mathcal R}}d$.
We have to consider two cases.
\begin{itemize}
\item In the former case, let $s\in h(a)$, $t\in h(b)$ and $U\subseteq W$ be such that $h(d)\subseteq U$ and $R(U)(s,t)$.
Hence, $a\in s$ and $b\in t$.
Since $(a,b)\vdash_{{\mathcal R}}d$ and $R(U)(s,t)$, therefore let $e\in{\mathcal R}$ be such that $d\not\leq_{{\mathcal R}}e$ and for all $u\in U$, $e\in u$.
Let $v\in W$ be such that $d\in v$ and $e\not\in v$.
Thus, $v\in h(d)$.
Since $h(d)\subseteq U$, therefore $v\in U$.
Since for all $u\in U$, $e\in u$, therefore $e\in v$: a contradiction.
\item In the latter case, let $s_{a}$ be the set of all regions $a^{\prime}$ in ${\mathcal R}$ such that $a\leq_{{\mathcal R}}a^{\prime}$ and $t_{b}$ be the set of all regions $b^{\prime}$ in ${\mathcal R}$ such that $b\leq_{{\mathcal R}}b^{\prime}$.
Remark that $a\in s_{a}$ and $b\in t_{b}$.
\begin{claim}\label{claim:universal:condition:on:sa:and:tb}
For all $a^{\prime}\in s_{a}$ and for all $b^{\prime}\in t_{b}$, $(a^{\prime},b^{\prime})\not\vdash_{{\mathcal R}}d$.
\end{claim}
\begin{claim}\label{claim:sa:and:tb:are:filters}
$s_{a}$ and $t_{b}$ are filters in the Boolean algebra $({\mathcal R},0_{{\mathcal R}},\star_{{\mathcal R}},\cup_{{\mathcal R}})$.
\end{claim}
For all filters $u,v$ in the Boolean algebra $({\mathcal R},0_{{\mathcal R}},\star_{{\mathcal R}},\cup_{{\mathcal R}})$, let $u^{l}$ be the set of all regions $b^{\prime}$ in ${\mathcal R}$ such that there exists $a^{\prime}\in u$ such that $(a^{\prime},b^{\prime})\vdash_{{\mathcal R}}d$ and $v^{r}$ be the set of all regions $a^{\prime}$ in ${\mathcal R}$ such that there exists $b^{\prime}\in v$ such that $(a^{\prime},b^{\prime})\vdash_{{\mathcal R}}d$.
\begin{claim}\label{claim:about:ul:and:vr}
For all filters $u,v$ in the Boolean algebra $({\mathcal R},0_{{\mathcal R}},\star_{{\mathcal R}},\cup_{{\mathcal R}})$, $u^{l}$ and $v^{r}$ are ideals in the Boolean algebra $({\mathcal R},0_{{\mathcal R}},\star_{{\mathcal R}},\cup_{{\mathcal R}})$.
\end{claim}
\begin{claim}\label{claim:about:three:equivalent:conditions}
For all filters $u,v$ in the Boolean algebra $({\mathcal R},0_{{\mathcal R}},\star_{{\mathcal R}},\cup_{{\mathcal R}})$, the following conditions are equivalent:
\begin{enumerate}
\item For all $a^{\prime}\in u$ and for all $b^{\prime}\in v$, $(a^{\prime},b^{\prime})\not\vdash_{{\mathcal R}}d$.
\item $u^{l}\cap v=\emptyset$.
\item $u\cap v^{r}=\emptyset$.
\end{enumerate}
\end{claim}
By Claims~\ref{claim:universal:condition:on:sa:and:tb}--\ref{claim:about:three:equivalent:conditions}, $s_{a}^{l}$ is an ideal in the Boolean algebra $({\mathcal R},0_{{\mathcal R}},\star_{{\mathcal R}},\cup_{{\mathcal R}})$, $t_{b}$ is a filter in the Boolean algebra $({\mathcal R},0_{{\mathcal R}},\star_{{\mathcal R}},\cup_{{\mathcal R}})$ and $s_{a}^{l}\cap t_{b}=\emptyset$.
By classical results in the theory of filters and ideals~\cite{Givant:Halmos:2009}, let $t$ be a maximal filter in the Boolean algebra $({\mathcal R},0_{{\mathcal R}},\star_{{\mathcal R}},\cup_{{\mathcal R}})$ such that $t_{b}\subseteq t$ and $s_{a}^{l}\cap t=\emptyset$.
Since $b\in t_{b}$, therefore $b\in t$ and $t\in h(b)$.
Moreover, by Claims~\ref{claim:sa:and:tb:are:filters}--\ref{claim:about:three:equivalent:conditions}, $s_{a}$ is a filter in the Boolean algebra $({\mathcal R},0_{{\mathcal R}},\star_{{\mathcal R}},\cup_{{\mathcal R}})$, $t^{r}$ is an ideal in the Boolean algebra $({\mathcal R},0_{{\mathcal R}},\star_{{\mathcal R}},\cup_{{\mathcal R}})$ and $s_{a}\cap t^{r}=\emptyset$.
By classical results in the theory of filters and ideals~\cite{Givant:Halmos:2009}, let $s$ be a maximal filter in the Boolean algebra $({\mathcal R},0_{{\mathcal R}},\star_{{\mathcal R}},\cup_{{\mathcal R}})$ such that $s_{a}\subseteq s$ and $s^{l}\cap t=\emptyset$.
Since $a\in s_{a}$, therefore $a\in s$ and $s\in h(a)$.
Moreover, since $t$ is a maximal filter in the Boolean algebra $({\mathcal R},0_{{\mathcal R}},\star_{{\mathcal R}},\cup_{{\mathcal R}})$, therefore by Claim~\ref{claim:about:three:equivalent:conditions}, for all $a^{\prime}\in s$ and for all $b^{\prime}\in t$, $(a^{\prime},b^{\prime})\not\vdash_{{\mathcal R}}d$.
Since $(h(a),h(b))\vdash_{W}h(d)$ and $t\in h(b)$, therefore not-$R(h(d))(s,t)$.
Let $a^{\prime\prime},b^{\prime\prime},d^{\prime\prime}\in{\mathcal R}$ be such that $a^{\prime\prime}\in s$, $b^{\prime\prime}\in t$, $(a^{\prime\prime},b^{\prime\prime})\vdash_{{\mathcal R}}d^{\prime\prime}$ and for all $e\in{\mathcal R}$, $d^{\prime\prime}\leq_{{\mathcal R}}e$, or there exists $u\in h(d)$ such that $e\not\in u$.
Hence, $d^{\prime\prime}\leq_{{\mathcal R}}d$, or there exists $u\in h(d)$ such that $d\not\in u$.
Since for all $u\in h(d)$, $d\in u$, therefore $d^{\prime\prime}\leq_{{\mathcal R}}d$.
Since $(a^{\prime\prime},b^{\prime\prime})\vdash_{{\mathcal R}}d^{\prime\prime}$, therefore $(a^{\prime\prime},b^{\prime\prime})\vdash_{{\mathcal R}}d$.
Since for all $a^{\prime}\in s$ and for all $b^{\prime}\in t$, $(a^{\prime},b^{\prime})\not\vdash_{{\mathcal R}}d$, therefore $a^{\prime\prime}\not\in s$, or $b^{\prime\prime}\not\in t$: a contradiction.
\end{itemize}
Hence, $(a,b)\vdash_{{\mathcal R}}d$ iff $(h(a),h(b))\vdash_{W}h(d)$.
\end{proof}
\\
\\
This completes the proof of Proposition~\ref{proposition:relational:representations:of:extended:contact:algebras}.
\end{proof}
\\
\\
The main drawback of the kind of extended contact algebra $({\mathcal P}(W),0_{W},\star_{W},\cup_{W},
$\linebreak$
\vdash_{W})$ considered in Proposition~\ref{proposition:relational:representations:of:extended:contact:algebras} is that the parametrized frame $(W,R)$ it is based on is a relatively complex relational structure.
In the next sections, we introduce two other kinds of extended contact algebra based on equivalence relations.
\section{Equivalence frames of type $1$}\label{subsection:Equivalence:frames:of:type:1}
An {\it equivalence frame of type $1$} is a structure of the form $(W,R)$ where $W$ is a nonempty set and $R$ is an equivalence relation on $W$.
In an equivalence frame $(W,R)$ of type $1$, the {\it equivalence class of $s\in W$ modulo $R$} will be denoted $R(s)$.
Given an equivalence frame $(W,R)$ of type $1$, let $\vdash_{W}$ be the ternary relation on $W$'s powerset defined by
\begin{itemize}
\item $(A,B)\vdash_{W}D$ iff the intersection of $A$ and $B$ is included in $D$ and for all $s\in W$, if $R(s)$ intersects both $A$ and $B$ then $R(s)$ intersects $D$.
\end{itemize}
The reader may easily verify that the structure $({\mathcal P}(W),0_{W},\star_{W},\cup_{W},\vdash_{W})$ is an extended contact algebra.
With the following proposition, one can say that these extended contact algebras are typical examples of extended contact algebras.
\begin{proposition}\label{proposition:relational:representations:of:type:1:of:finite:extended:contact:algebras}
Let $({\mathcal R},0_{{\mathcal R}},\star_{{\mathcal R}},\cup_{{\mathcal R}},\vdash_{{\mathcal R}})$ be a finite extended contact algebra.
There exists a finite equivalence frame $(W,R)$ of type $1$ and an embedding of $({\mathcal R},0_{{\mathcal R}},\star_{{\mathcal R}},\cup_{{\mathcal R}},
$\linebreak$
\vdash_{{\mathcal R}})$ in $({\mathcal P}(W),0_{W},\star_{W},\cup_{W},\vdash_{W})$.
\end{proposition}
\begin{proof}
By Proposition~\ref{proposition:topological:representations:of:extended:contact:algebras}, let $(X,\tau)$ be a finite topological space and $h$ be a surjective embedding of $({\mathcal R},0_{{\mathcal R}},\star_{{\mathcal R}},\cup_{{\mathcal R}},\vdash_{{\mathcal R}})$ in $(RC(X,\tau),0_{X},\star_{X},\cup_{X},\vdash_{X})$.
Hence, the Boolean algebra $(RC(X,\tau),0_{X},\star_{X},\cup_{X})$ of all regular closed subsets of $X$ is finite.
Let $(W,R)$ be the structure where
\begin{itemize}
\item $W$ is the set of all couples of the form $(A,s)$ where $A\in RC(X,\tau)$ and $s\in X$ are such that $A$ is an atom of $(RC(X,\tau),0_{X},\star_{X},\cup_{X})$ and $s\in A$,
\item $R$ is the binary relation on $W$ defined by $R((A,s),(B,t))$ iff $s=t$.
\end{itemize}
Obviously, $(W,R)$ is an equivalence frame of type $1$.
Let $h^{\prime}$ be the function associating to each region in ${\mathcal R}$ a subset of $W$ defined by
\begin{itemize}
\item $h^{\prime}(a)$ is the set of all $(A,s)\in W$ such that $A\subseteq h(a)$.
\end{itemize}
In order to prove that $h^{\prime}$ is an embedding of $({\mathcal R},0_{{\mathcal R}},\star_{{\mathcal R}},\cup_{{\mathcal R}},\vdash_{{\mathcal R}})$ in $({\mathcal P}(W),0_{W},\star_{W},
$\linebreak$
\cup_{W},\vdash_{W})$, let us prove the following
\begin{lemma}\label{lemma:for:the:proposition:relational:representations:of:type:1:of:finite:extended:contact:algebras}
The following conditions hold:
\begin{enumerate}
\item $h^{\prime}$ is injective.
\item $h^{\prime}(0_{{\mathcal R}})=0_{W}$.
\item For all regions $a$ in ${\mathcal R}$, $h^{\prime}(a^{\star_{{\mathcal R}}})=h^{\prime}(a)^{\star_{W}}$.
\item For all regions $a,b$ in ${\mathcal R}$, $h^{\prime}(a\cup_{{\mathcal R}}b)=h^{\prime}(a)\cup_{W}h^{\prime}(b)$.
\item For all regions $a,b,d$ in ${\mathcal R}$, $(a,b)\vdash_{{\mathcal R}}d$ iff $(h^{\prime}(a),h^{\prime}(b))\vdash_{W}h^{\prime}(d)$.
\end{enumerate}
\end{lemma}
\begin{proof}
$(1)$~We demonstrate $h^{\prime}$ is injective.
Let $a,b$ be arbitrary distinct regions in ${\mathcal R}$.
Since $h$ is a surjective embedding of $({\mathcal R},0_{{\mathcal R}},\star_{{\mathcal R}},\cup_{{\mathcal R}},\vdash_{{\mathcal R}})$ in $(RC(X,\tau),0_{X},\star_{X},\cup_{X},
$\linebreak$
\vdash_{X})$, therefore $h(a)$ and $h(b)$ are distinct regular closed subsets of $X$.
Hence, $h(a)\not\subseteq h(b)$, or $h(b)\not\subseteq h(a)$.
Without loss of generality, suppose $h(a)\not\subseteq h(b)$.
Let $s\in X$ be such that $s\in h(a)$ and $s\not\in h(b)$.
Since the Boolean algebra $(RC(X,\tau),0_{X},\star_{X},\cup_{X})$ of all regular closed subsets of $X$ is finite, therefore let $n$ be a nonnegative integer and $A_{1},\ldots,A_{n}$ be atoms of $(RC(X,\tau),0_{X},\star_{X},\cup_{X})$ such that $h(a)=A_{1}\cup_{X}\ldots\cup_{X}A_{n}$.
Since $s\in h(a)$, therefore let $i\leq n$ be a positive integer such that $s\in A_{i}$.
Since $A_{i}$ is an atom of $(RC(X,\tau),0_{X},\star_{X},\cup_{X})$, therefore the couple $(A_{i},s)$ is in $W$.
Since $A_{i}\subseteq h(a)$, therefore $(A_{i},s)\in h^{\prime}(a)$.
Since $s\not\in h(b)$ and $s\in A_{i}$, therefore $A_{i}\not\subseteq h(b)$.
Thus, $(A_{i},s)\not\in h^{\prime}(b)$.
Since $(A_{i},s)\in h^{\prime}(a)$, therefore $h^{\prime}(a)\not\subseteq h^{\prime}(b)$.
Consequently, $h^{\prime}(a)$ and $h^{\prime}(b)$ are distinct subsets of $W$.
Since $a,b$ were arbitrary, therefore $h^{\prime}$ is injective.
\\
\\
$(2)$~We demonstrate $h^{\prime}(0_{{\mathcal R}})=0_{W}$.
Suppose $h^{\prime}(0_{{\mathcal R}})\not=0_{W}$.
Let $(A,s)$ be a couple in $W$ such that $(A,s)\in h^{\prime}(0_{{\mathcal R}})$.
Hence, $A$ is an atom of $(RC(X,\tau),0_{X},\star_{X},\cup_{X})$.
Moreover, $A\subseteq h(0_{{\mathcal R}})$.
Since $h$ is a surjective embedding of $({\mathcal R},0_{{\mathcal R}},\star_{{\mathcal R}},\cup_{{\mathcal R}},\vdash_{{\mathcal R}})$ in $(RC(X,\tau),0_{X},\star_{X},\cup_{X},\vdash_{X})$, therefore $h(0_{{\mathcal R}})=0_{X}$.
Since $A\subseteq h(0_{{\mathcal R}})$, therefore $A\subseteq0_{X}$.
Thus, $A$ is not an atom: a contradiction.
Consequently, $h^{\prime}(0_{{\mathcal R}})=0_{W}$.
\\
\\
$(3)$~Let $a$ be a region in ${\mathcal R}$
We demonstrate $h^{\prime}(a^{\star_{{\mathcal R}}})=h^{\prime}(a)^{\star_{W}}$.
Suppose $h^{\prime}(a^{\star_{{\mathcal R}}})\not=h^{\prime}(a)^{\star_{W}}$.
Hence, $h^{\prime}(a^{\star_{{\mathcal R}}})\not\subseteq h^{\prime}(a)^{\star_{W}}$, or $h^{\prime}(a)^{\star_{W}}\not\subseteq h^{\prime}(a^{\star_{{\mathcal R}}})$.
We have to consider two cases.
\begin{itemize}
\item In the former case, let $(A,s)$ be a couple in $W$ such that $(A,s)\in h^{\prime}(a^{\star_{{\mathcal R}}})$ and $(A,s)\not\in h^{\prime}(a)^{\star_{W}}$.
Thus, $A$ is an atom of $(RC(X,\tau),0_{X},\star_{X},\cup_{X})$.
Moreover, $A\subseteq h(a^{\star_{{\mathcal R}}})$.
Since $h$ is a surjective embedding of $({\mathcal R},0_{{\mathcal R}},\star_{{\mathcal R}},\cup_{{\mathcal R}},\vdash_{{\mathcal R}})$ in $(RC(X,\tau),0_{X},\star_{X},\cup_{X},\vdash_{X})$, therefore $h(a^{\star_{{\mathcal R}}})=h(a)^{\star_{X}}$.
Since $A\subseteq h(a^{\star_{{\mathcal R}}})$, therefore $A\subseteq h(a)^{\star_{X}}$.
Since $(A,s)\not\in h^{\prime}(a)^{\star_{W}}$, therefore $(A,s)\in h^{\prime}(a)$.
Consequently, $A\subseteq h(a)$.
Since $A\subseteq h(a)^{\star_{X}}$, therefore $A$ is not an atom: a contradiction.
\item In the latter case, let $(A,s)$ be a couple in $W$ such that $(A,s)\in h^{\prime}(a)^{\star_{W}}$ and $(A,s)\not\in h^{\prime}(a^{\star_{{\mathcal R}}})$.
Hence, $A$ is an atom of $(RC(X,\tau),0_{X},\star_{X},\cup_{X})$.
Moreover, $(A,s)\not\in h^{\prime}(a)$.
Thus, $A\not\subseteq h(a)$.
Since $(A,s)\not\in h^{\prime}(a^{\star_{{\mathcal R}}})$, therefore $A\not\subseteq h(a^{\star_{{\mathcal R}}})$.
Since $h$ is a surjective embedding of $({\mathcal R},0_{{\mathcal R}},\star_{{\mathcal R}},\cup_{{\mathcal R}},\vdash_{{\mathcal R}})$ in $(RC(X,\tau),0_{X},\star_{X},\cup_{X},\vdash_{X})$, therefore $h(a^{\star_{{\mathcal R}}})=h(a)^{\star_{X}}$.
Since $A\not\subseteq h(a^{\star_{{\mathcal R}}})$, therefore $A\not\subseteq h(a)^{\star_{X}}$.
Since $A\not\subseteq h(a)$, therefore $A$ is not an atom: a contradiction.
\end{itemize}
Consequently, $h^{\prime}(a^{\star_{{\mathcal R}}})=h^{\prime}(a)^{\star_{W}}$.
\\
\\
$(4)$~Let $a,b$ be regions in ${\mathcal R}$.
We demonstrate $h^{\prime}(a\cup_{{\mathcal R}}b)=h^{\prime}(a)\cup_{W}h^{\prime}(b)$.
Suppose $h^{\prime}(a\cup_{{\mathcal R}}b)\not=h^{\prime}(a)\cup_{W}h^{\prime}(b)$.
Hence, $h^{\prime}(a\cup_{{\mathcal R}}b)\not\subseteq h^{\prime}(a)\cup_{W}h^{\prime}(b)$, or $h^{\prime}(a)\cup_{W}h^{\prime}(b)\not\subseteq h^{\prime}(a\cup_{{\mathcal R}}b)$.
We have to consider two cases.
\begin{itemize}
\item In the former case, let $(A,s)$ be a couple in $W$such that $(A,s)\in h^{\prime}(a\cup_{{\mathcal R}}b)$ and $(A,s)\not\in h^{\prime}(a)\cup_{W}h^{\prime}(b)$.
Thus, $A$ is an atom of $(RC(X,\tau),0_{X},\star_{X},\cup_{X})$.
Moreover, $A\subseteq h(a\cup_{{\mathcal R}}b)$.
Since $h$ is a surjective embedding of $({\mathcal R},0_{{\mathcal R}},\star_{{\mathcal R}},\cup_{{\mathcal R}},\vdash_{{\mathcal R}})$ in $(RC(X,\tau),0_{X},\star_{X},\cup_{X},\vdash_{X})$, therefore $h(a\cup_{{\mathcal R}}b)=h(a)\cup_{X}h(b)$.
Since $A\subseteq h(a\cup_{{\mathcal R}}b)$, therefore $A\subseteq h(a)\cup_{X}h(b)$.
Since $(A,s)\not\in h^{\prime}(a)\cup_{W}h^{\prime}(b)$, therefore $(A,s)\not\in h^{\prime}(a)$ and $(A,s)\not\in h^{\prime}(b)$.
Consequently, $A\not\subseteq h(a)$ and $A\not\subseteq h(b)$.
Since $A\subseteq h(a)\cup_{X}h(b)$, therefore $A$ is not an atom: a contradiction.
\item In the latter case, let $(A,s)$ be a couple in $W$such that $(A,s)\in h^{\prime}(a)\cup_{W}h^{\prime}(b)$ and $(A,s)\not\in h^{\prime}(a\cup_{{\mathcal R}}b)$.
Hence, $A$ is an atom of $(RC(X,\tau),0_{X},\star_{X},\cup_{X})$.
Moreover, $(A,s)\in h^{\prime}(a)$, or $(A,s)\in h^{\prime}(b)$.
Thus, $(A,s)\in h^{\prime}(a)$, or $(A,s)\in h^{\prime}(b)$.
Consequently, $A\subseteq h(a)$, or $A\subseteq h(b)$.
Hence, $A\subseteq h(a)\cup_{X}h(b)$.
Since $(A,s)\not\in h^{\prime}(a\cup_{{\mathcal R}}b)$, therefore $A\not\subseteq h(a\cup_{{\mathcal R}}b)$.
Since $h$ is a surjective embedding of $({\mathcal R},0_{{\mathcal R}},\star_{{\mathcal R}},\cup_{{\mathcal R}},\vdash_{{\mathcal R}})$ in $(RC(X,\tau),0_{X},\star_{X},\cup_{X},\vdash_{X})$, therefore $h(a\cup_{{\mathcal R}}b)=h(a)\cup_{X}h(b)$.
Since $A\not\subseteq h(a\cup_{{\mathcal R}}b)$, therefore $A\not\subseteq h(a)\cup_{X}h(b)$: a contradiction.
\end{itemize}
Thus, $h^{\prime}(a\cup_{{\mathcal R}}b)=h^{\prime}(a)\cup_{W}h^{\prime}(b)$.
\\
\\
$(5)$~Let $a,b,d$ be regions in ${\mathcal R}$.
We demonstrate $(a,b)\vdash_{{\mathcal R}}d$ iff $(h^{\prime}(a),h^{\prime}(b))\vdash_{W}h^{\prime}(d)$.
Suppose $(a,b)\vdash_{{\mathcal R}}d$ not-iff $(h^{\prime}(a),h^{\prime}(b))\vdash_{W}h^{\prime}(d)$.
Hence, $(a,b)\vdash_{{\mathcal R}}d$ and $(h^{\prime}(a),h^{\prime}(b))\not\vdash_{W}h^{\prime}(d)$, or $(h^{\prime}(a),h^{\prime}(b))\vdash_{W}h^{\prime}(d)$ and $(a,b)\not\vdash_{{\mathcal R}}d$.
We have to consider two cases.
\begin{itemize}
\item In the former case, since $h$ is a surjective embedding of $({\mathcal R},0_{{\mathcal R}},\star_{{\mathcal R}},\cup_{{\mathcal R}},\vdash_{{\mathcal R}})$ in $(RC(X,\tau),0_{X},\star_{X},\cup_{X},\vdash_{X})$, therefore $(h(a),h(b))\vdash_{X}h(d)$.
Thus, $h(a)\cap h(b)\subseteq h(d)$.
Since $(h^{\prime}(a),h^{\prime}(b))\not\vdash_{W}h^{\prime}(d)$, therefore $h^{\prime}(a)\cap h^{\prime}(b)\not\subseteq h^{\prime}(d)$, or there exists a couple $(E,w)$ in $W$ such that $R((E,w))\cap h^{\prime}(a)\not=\emptyset$, $R((E,w))\cap h^{\prime}(b)\not=\emptyset$ and $R((E,w))\cap h^{\prime}(d)=\emptyset$.
We have to consider two subcases.
\begin{itemize}
\item In the former subcase, let $(E,w)$ be a couple in $W$ such that $(E,w)\in h^{\prime}(a)$, $(E,w)\in h^{\prime}(b)$ and $(E,w)\not\in h^{\prime}(d)$.
Consequently, $E$ is an atom of $(RC(X,\tau),0_{X},\star_{X},\cup_{X})$.
Moreover, $E\subseteq h(a)$, $E\subseteq h(b)$ and $E\not\subseteq h(d)$.
Hence, $E\subseteq h(a)\cap h(b)$.
Since $h(a)\cap h(b)\subseteq h(d)$, therefore $E\subseteq h(d)$: a contradiction.
\item In the latter subcase, let $(E,w)$ be a couple in $W$ such that $R((E,w))\cap h^{\prime}(a)\not=\emptyset$, $R((E,w))\cap h^{\prime}(b)\not=\emptyset$ and $R((E,w))\cap h^{\prime}(d)=\emptyset$.
Let $(A,s)$ and $(B,t)$ be couples in $W$ such that $R((E,w),(A,s))$, $(A,s)\in h^{\prime}(a)$, $R((E,w),(B,t))$ and $(B,t)\in h^{\prime}(b)$.
Thus, $A$ and $B$ are atoms of $(RC(X,\tau),0_{X},\star_{X},\cup_{X})$ such that $s\in A$ and $t\in B$.
Moreover, $w=s$, $A\subseteq h(a)$, $w=t$ and $B\subseteq h(b)$.
Consequently, $w\in h(a)\cap h(b)$.
Since $h(a)\cap h(b)\subseteq h(d)$, therefore $w\in h(d)$.
Since the Boolean algebra $(RC(X,\tau),0_{X},\star_{X},\cup_{X})$ of all regular closed subsets of $X$ is finite, therefore let $n$ be a nonnegative integer and $D_{1},\ldots,D_{n}$ be atoms of $(RC(X,\tau),0_{X},\star_{X},\cup_{X})$ such that $h(d)=D_{1}\cup_{X}\ldots\cup_{X}D_{n}$.
Since $w\in h(d)$, therefore let $i\leq n$ be a positive integer such that $w\in D_{i}$.
Hence, $(D_{i},w)$ is a couple in $W$.
Moreover, $(D_{i},w)\in R((E,w))$ and $D_{i}\subseteq h(d)$.
Thus, $R((E,w))\cap h^{\prime}(d)\not=\emptyset$: a contradiction.
\end{itemize}
\item In the latter case, since $h$ is a surjective embedding of $({\mathcal R},0_{{\mathcal R}},\star_{{\mathcal R}},\cup_{{\mathcal R}},\vdash_{{\mathcal R}})$ in $(RC(X,\tau),0_{X},\star_{X},\cup_{X},\vdash_{X})$, therefore $(h(a),h(b))\not\vdash_{X}h(d)$.
Consequently, $h(a)\cap h(b)\not\subseteq h(d)$.
Let $w\in X$ be such that $w\in h(a)$, $w\in h(b)$ and $w\not\in h(d)$.
Since the Boolean algebra $(RC(X,\tau),0_{X},\star_{X},\cup_{X})$ of all regular closed subsets of $X$ is finite, therefore let $m$ and $n$ be nonnegative integers and $A_{1},\ldots,A_{m}$ and $B_{1},\ldots,B_{n}$ be atoms of $(RC(X,\tau),0_{X},\star_{X},\cup_{X})$ such that $h(a)=A_{1}\cup_{X}\ldots\cup_{X}A_{m}$ and $h(b)=B_{1}\cup_{X}\ldots\cup_{X}B_{n}$.
Since $w\in h(a)$ and $w\in h(b)$, therefore let $i\leq m$ and $j\leq n$ be positive integers such that $w\in A_{i}$ and $w\in B_{j}$.
Since $A_{i}$ and $B_{j}$ are atoms in $(RC(X,\tau),0_{X},\star_{X},\cup_{X})$, therefore the couples $(A_{i},w)$ and $(B_{j},w)$ are in $W$.
Since $A_{i}\subseteq h(a)$ and $B_{j}\subseteq h(b)$, therefore $(A_{i},w)\in h^{\prime}(a)$ and $(B_{j},w)\in h^{\prime}(b)$.
Since the Boolean algebra $(RC(X,\tau),0_{X},\star_{X},\cup_{X})$ of all regular closed subsets of $X$ is finite, therefore let $E$ be an atom of $(RC(X,\tau),0_{X},\star_{X},\cup_{X})$ such that $w\in E$.
Hence, the couple $(E,w)$ is in $W$.
Moreover, $R((E,w),(A_{i},w))$ and $R((E,w),(B_{j},w))$.
Since $(A_{i},w)\in h^{\prime}(a)$ and $(B_{j},w)\in h^{\prime}(b)$, therefore $R((E,w))\cap h^{\prime}(a)\not=\emptyset$ and $R((E,w))\cap h^{\prime}(b)\not=\emptyset$.
Since $(h^{\prime}(a),h^{\prime}(b))\vdash_{W}h^{\prime}(d)$, therefore $R((E,w))\cap h^{\prime}(d)\not=\emptyset$.
Let $(E^{\prime},w^{\prime})$ be a couple in $W$ such that $R((E,w),(E^{\prime},w^{\prime}))$ and $(E^{\prime},w^{\prime})\in h^{\prime}(d)$.
Thus, $w=w^{\prime}$, $w^{\prime}\in E^{\prime}$ and $E^{\prime}\subseteq h(d)$.
Consequently, $w\in h(d)$: a contradiction.
\end{itemize}
Hence, $(a,b)\vdash_{{\mathcal R}}d$ iff $(h^{\prime}(a),h^{\prime}(b))\vdash_{W}h^{\prime}(d)$.
\end{proof}
\\
\\
This completes the proof of Proposition~\ref{proposition:relational:representations:of:type:1:of:finite:extended:contact:algebras}.
\end{proof}
\section{Equivalence frames of type $2$}\label{subsection:Equivalence:frames:of:type:2}
The weak point of Proposition~\ref{proposition:relational:representations:of:type:1:of:finite:extended:contact:algebras} is that it does not say whether the embedding preserves the relation of internal connectedness.
In this section, we introduce another type of equivalence frames with which we will be able to embed any finite extended contact algebra while preserving its relation of internal connectedness.
An {\it equivalence frame of type $2$} is a structure of the form $(W,R_{1},R_{2})$ where $W$ is a nonempty set and $R_{1}$ and $R_{2}$ are equivalence relations on $W$.
In an equivalence frame $(W,R_{1},R_{2})$ of type $2$, the {\it equivalence class of $s\in W$ modulo $R_{1}$} will be denoted $R_{1}(s)$ and the {\it equivalence class of $s\in W$ modulo $R_{2}$} will be denoted $R_{2}(s)$.
Moreover, for all $s\in W$, $R_{1}(R_{2}(s))$ will denote the union of all $R_{1}(t)$ when $t$ ranges over $R_{2}(s)$.
Given an equivalence frame $(W,R_{1},R_{2})$ of type $2$, let $\vdash_{W}$ be the ternary relation on $W$'s powerset defined by
\begin{itemize}
\item $(A,B)\vdash_{W}D$ iff the intersection of $A$ and $B$ is included in $D$ and for all $s\in W$, if $R_{1}(R_{2}(s))$ intersects both $A$ and $B$ then $R_{1}(R_{2}(s))$ intersects $D$.
\end{itemize}
The reader may easily verify that the structure $({\mathcal P}(W),0_{W},\star_{W},\cup_{W},\vdash_{W})$ is an extended contact algebra.
With the following proposition, one can say that these extended contact algebras are typical examples of extended contact algebras.
\begin{proposition}\label{proposition:relational:representations:of:type:2:of:finite:extended:contact:algebras}
Let $({\mathcal R},0_{{\mathcal R}},\star_{{\mathcal R}},\cup_{{\mathcal R}},\vdash_{{\mathcal R}})$ be a finite extended contact algebra.
There exists an equivalence frame $(W,R_{1},R_{2})$ of type $2$ and an embedding of $({\mathcal R},0_{{\mathcal R}},\star_{{\mathcal R}},
$\linebreak$
\cup_{{\mathcal R}},\vdash_{{\mathcal R}})$ in $({\mathcal P}(W),0_{W},\star_{W},\cup_{W},\vdash_{W})$ preserving the relation of internal connectedness.
\end{proposition}
\begin{proof}
By Proposition~\ref{proposition:topological:representations:of:extended:contact:algebras}, let $(X,\tau)$ be a topological space and $h$ be a surjective embedding of $({\mathcal R},0_{{\mathcal R}},\star_{{\mathcal R}},\cup_{{\mathcal R}},\vdash_{{\mathcal R}})$ in $(RC(X,\tau),0_{X},\star_{X},\cup_{X},\vdash_{X})$.
As proved in~\cite[Chapter~$2$]{Ivanova:2016}, the topological space $(X,\tau)$ is finite.
Hence, the Boolean algebra $(RC(X,
$\linebreak$
\tau),0_{X},\star_{X},\cup_{X})$ of all regular closed subsets of $X$ is finite.
Let $(W,R_{1},R_{2})$ be the structure where
\begin{itemize}
\item $W$ is the set of all couples of the form $(A,s)$ where $A\in RC(X,\tau)$ and $s\in X$ are such that $A$ is an atom of $(RC(X,\tau),0_{X},\star_{X},\cup_{X})$ and $s\in A$,
\item $R_{1}$ is the binary relation on $W$ defined by $R_{1}((A,s),(B,t))$ iff $A=B$,
\item $R_{2}$ is the binary relation on $W$ defined by $R_{2}((A,s),(B,t))$ iff $s=t$.
\end{itemize}
Obviously, $(W,R_{1},R_{2})$ is an equivalence frame of type $2$.
Let $h^{\prime}$ be the function associating to each region in ${\mathcal R}$ a subset of $W$ defined by
\begin{itemize}
\item $h^{\prime}(a)$ is the set of all $(A,s)\in W$ such that $A\subseteq h(a)$.
\end{itemize}
In order to prove that $h^{\prime}$ is an embedding of $({\mathcal R},0_{{\mathcal R}},\star_{{\mathcal R}},\cup_{{\mathcal R}},\vdash_{{\mathcal R}})$ in $({\mathcal P}(W),0_{W},\star_{W},
$\linebreak$
\cup_{W},\vdash_{W})$, let us prove the following
\begin{lemma}\label{first:lemma:for:the:proposition:relational:representations:of:type:2:of:finite:extended:contact:algebras}
The following conditions hold:
\begin{enumerate}
\item $h^{\prime}$ is injective.
\item $h^{\prime}(0_{{\mathcal R}})=0_{W}$.
\item For all regions $a$ in ${\mathcal R}$, $h^{\prime}(a^{\star_{{\mathcal R}}})=h^{\prime}(a)^{\star_{W}}$.
\item For all regions $a,b$ in ${\mathcal R}$, $h^{\prime}(a\cup_{{\mathcal R}}b)=h^{\prime}(a)\cup_{W}h^{\prime}(b)$.
\item For all regions $a,b,d$ in ${\mathcal R}$, $(a,b)\vdash_{{\mathcal R}}d$ iff $(h^{\prime}(a),h^{\prime}(b))\vdash_{W}h^{\prime}(d)$.
\end{enumerate}
\end{lemma}
\begin{proof}
The proofs of items~$(1)$--$(4)$ are similar to the proofs of the corresponding items in Lemma~\ref{lemma:for:the:proposition:relational:representations:of:type:1:of:finite:extended:contact:algebras}.
\\
\\
$(5)$~Let $a,b,d$ be regions in ${\mathcal R}$.
We demonstrate $(a,b)\vdash_{{\mathcal R}}d$ iff $(h^{\prime}(a),h^{\prime}(b))\vdash_{W}h^{\prime}(d)$.
Suppose $(a,b)\vdash_{{\mathcal R}}d$ not-iff $(h^{\prime}(a),h^{\prime}(b))\vdash_{W}h^{\prime}(d)$.
Hence, $(a,b)\vdash_{{\mathcal R}}d$ and $(h^{\prime}(a),h^{\prime}(b))\not\vdash_{W}h^{\prime}(d)$, or $(h^{\prime}(a),h^{\prime}(b))\vdash_{W}h^{\prime}(d)$ and $(a,b)\not\vdash_{{\mathcal R}}d$.
We have to consider two cases.
\begin{itemize}
\item In the former case, since $h$ is a surjective embedding of $({\mathcal R},0_{{\mathcal R}},\star_{{\mathcal R}},\cup_{{\mathcal R}},\vdash_{{\mathcal R}})$ in $(RC(X,\tau),0_{X},\star_{X},\cup_{X},\vdash_{X})$, therefore $(h(a),h(b))\vdash_{X}h(d)$.
Thus, $h(a)\cap h(b)\subseteq h(d)$.
Since $(h^{\prime}(a),h^{\prime}(b))\not\vdash_{W}h^{\prime}(d)$, therefore $h^{\prime}(a)\cap h^{\prime}(b)\not\subseteq h^{\prime}(d)$, or there exists a couple $(E,w)$ in $W$ such that $R_{1}(R_{2}((E,w)))\cap h^{\prime}(a)\not=\emptyset$, $R_{1}(R_{2}((E,w)))\cap h^{\prime}(b)\not=\emptyset$ and $R_{1}(R_{2}((E,w)))\cap h^{\prime}(d)=\emptyset$.
We have to consider two subcases.
\begin{itemize}
\item In the former subcase, let $(E,w)$ be a couple in $W$ such that $(E,w)\in h^{\prime}(a)$, $(E,w)\in h^{\prime}(b)$ and $(E,w)\not\in h^{\prime}(d)$.
Consequently, $E$ is an atom of $(RC(X,\tau),0_{X},\star_{X},\cup_{X})$.
Moreover, $E\subseteq h(a)$, $E\subseteq h(b)$ and $E\not\subseteq h(d)$.
Hence, $E\subseteq h(a)\cap h(b)$.
Since $h(a)\cap h(b)\subseteq h(d)$, therefore $E\subseteq h(d)$: a contradiction.
\item In the latter subcase, let $(E,w)$ be a couple in $W$ such that $R_{1}(R_{2}((E,w)))
$\linebreak$
\cap h^{\prime}(a)\not=\emptyset$, $R_{1}(R_{2}((E,w)))\cap h^{\prime}(b)\not=\emptyset$ and $R_{1}(R_{2}((E,w)))\cap h^{\prime}(d)=\emptyset$.
Let $(A,s)$, $(A^{\prime},s^{\prime})$, $(B,t)$ and $(B^{\prime},t^{\prime})$ be couples in $W$ such that $R_{2}((E,w),(A^{\prime},s^{\prime}))$, $R_{1}((A^{\prime},s^{\prime}),(A,s))$, $(A,s)\in h^{\prime}(a)$, $R_{2}((E,w),(B^{\prime},
$\linebreak$
t^{\prime}))$, $R_{1}((B^{\prime},t^{\prime}),(B,t))$ and $(B,t)\in h^{\prime}(b)$.
Thus, $A$, $A^{\prime}$, $B$ and $B^{\prime}$ are atoms of $(RC(X,\tau),0_{X},\star_{X},\cup_{X})$ such that $s\in A$, $s^{\prime}\in A^{\prime}$, $t\in B$ and $t^{\prime}\in B^{\prime}$.
Moreover, $w=s^{\prime}$, $A^{\prime}=A$, $A\subseteq h(a)$, $w=t^{\prime}$, $B^{\prime}=B$ and $B\subseteq h(b)$.
Consequently, $w\in h(a)\cap h(b)$.
Since $h(a)\cap h(b)\subseteq h(d)$, therefore $w\in h(d)$.
Since the Boolean algebra $(RC(X,\tau),0_{X},\star_{X},\cup_{X})$ of all regular closed subsets of $X$ is finite, therefore let $n$ be a nonnegative integer and $D_{1},\ldots,D_{n}$ be atoms of $(RC(X,\tau),0_{X},\star_{X},\cup_{X})$ such that $h(d)=D_{1}\cup_{X}\ldots\cup_{X}D_{n}$.
Since $w\in h(d)$, therefore let $i\leq n$ be a positive integer such that $w\in D_{i}$.
Hence, $(D_{i},w)$ is a couple in $W$.
Moreover, $(D_{i},w)\in R_{1}(R_{2}((E,w)))$ and $D_{i}\subseteq h(d)$.
Thus, $R_{1}(R_{2}((E,w)))\cap h^{\prime}(d)\not=\emptyset$: a contradiction.
\end{itemize}
\item In the latter case, since $h$ is a surjective embedding of $({\mathcal R},0_{{\mathcal R}},\star_{{\mathcal R}},\cup_{{\mathcal R}},\vdash_{{\mathcal R}})$ in $(RC(X,\tau),0_{X},\star_{X},\cup_{X},\vdash_{X})$, therefore $(h(a),h(b))\not\vdash_{X}h(d)$.
Consequently, $h(a)\cap h(b)\not\subseteq h(d)$.
Let $w\in X$ be such that $w\in h(a)$, $w\in h(b)$ and $w\not\in h(d)$.
Since the Boolean algebra $(RC(X,\tau),0_{X},\star_{X},\cup_{X})$ of all regular closed subsets of $X$ is finite, therefore let $m$ and $n$ be nonnegative integers and $A_{1},\ldots,A_{m}$ and $B_{1},\ldots,B_{n}$ be atoms of $(RC(X,\tau),0_{X},\star_{X},\cup_{X})$ such that $h(a)=A_{1}\cup_{X}\ldots\cup_{X}A_{m}$ and $h(b)=B_{1}\cup_{X}\ldots\cup_{X}B_{n}$.
Since $w\in h(a)$ and $w\in h(b)$, therefore let $i\leq m$ and $j\leq n$ be positive integers such that $w\in A_{i}$ and $w\in B_{j}$.
Since $A_{i}$ and $B_{j}$ are atoms in $(RC(X,\tau),0_{X},\star_{X},\cup_{X})$, therefore the couples $(A_{i},w)$ and $(B_{j},w)$ are in $W$.
Since $A_{i}\subseteq h(a)$ and $B_{j}\subseteq h(b)$, therefore $(A_{i},w)\in h^{\prime}(a)$ and $(B_{j},w)\in h^{\prime}(b)$.
Since the Boolean algebra $(RC(X,\tau),0_{X},\star_{X},\cup_{X})$ of all regular closed subsets of $X$ is finite, therefore let $E$ be an atom of $(RC(X,\tau),0_{X},\star_{X},\cup_{X})$ such that $w\in E$.
Hence, the couple $(E,w)$ is in $W$.
Moreover, $(A_{i},w)\in R_{1}(R_{2}((E,w)))$ and $(B_{j},w)\in R_{1}(R_{2}((E,w)))$.
Since $(A_{i},w)\in h^{\prime}(a)$ and $(B_{j},w)\in h^{\prime}(b)$, therefore
\linebreak$
R_{1}(R_{2}((E,w)))\cap h^{\prime}(a)\not=\emptyset$ and $R_{1}(R_{2}((E,w)))\cap h^{\prime}(b)\not=\emptyset$.
Since $(h^{\prime}(a),
$\linebreak$
h^{\prime}(b))\vdash_{W}h^{\prime}(d)$, therefore $R_{1}(R_{2}((E,w)))\cap h^{\prime}(d)\not=\emptyset$.
Let $(E^{\prime},w^{\prime})$ and $(E^{\prime\prime},w^{\prime\prime})$ be couples in $W$ such that $R_{2}((E,w),(E^{\prime\prime},w^{\prime\prime}))$, $R_{1}((E^{\prime\prime},w^{\prime\prime}),(E^{\prime},
$\linebreak$
w^{\prime}))$ and $(E^{\prime},w^{\prime})\in h^{\prime}(d)$.
Thus, $w=w^{\prime\prime}$, $w^{\prime\prime}\in E^{\prime\prime}$, $E^{\prime\prime}=E^{\prime}$ and $E^{\prime}\subseteq h(d)$.
Consequently, $w\in h(d)$: a contradiction.
\end{itemize}
Hence, $(a,b)\vdash_{{\mathcal R}}d$ iff $(h^{\prime}(a),h^{\prime}(b))\vdash_{W}h^{\prime}(d)$.
\end{proof}
\\
\\
Now, let us prove the following
\begin{lemma}\label{second:lemma:for:the:proposition:relational:representations:of:type:2:of:finite:extended:contact:algebras}
For all regions $a$ in ${\mathcal R}$, $c^{\circ}_{{\mathcal R}}(a)$ iff $c^{\circ}_{W}(h^{\prime}(a))$.
\end{lemma}
\begin{proof}
Let $a$ be a region in ${\mathcal R}$.
We demonstrate $c^{\circ}_{{\mathcal R}}(a)$ iff $c^{\circ}_{W}(h^{\prime}(a))$.
Suppose $c^{\circ}_{{\mathcal R}}(a)$ not-iff $c^{\circ}_{W}(h^{\prime}(a))$.
Hence, $c^{\circ}_{{\mathcal R}}(a)$ and not-$c^{\circ}_{W}(h^{\prime}(a))$, or not-$c^{\circ}_{{\mathcal R}}(a)$ and $c^{\circ}_{W}(h^{\prime}(a))$.
\begin{itemize}
\item In the former case, let $A_{1}^{\prime},A_{2}^{\prime}$ be subsets of $W$, $A_{1}^{\prime},A_{2}^{\prime}\not=0_{W}$, such that $h^{\prime}(a)=A_{1}^{\prime}\cup_{W}A_{2}^{\prime}$ and $(A_{1}^{\prime},A_{2}^{\prime})\vdash_{W}h^{\prime}(a)^{\star_{W}}$.
By item~$(3)$ of Lemma~\ref{first:lemma:for:the:proposition:relational:representations:of:type:2:of:finite:extended:contact:algebras}, $h^{\prime}(a^{\star_{{\mathcal R}}})=h^{\prime}(a)^{\star_{W}}$.
Since $(A_{1}^{\prime},A_{2}^{\prime})\vdash_{W}h^{\prime}(a)^{\star_{W}}$, therefore $(A_{1}^{\prime},A_{2}^{\prime})\vdash_{W}h^{\prime}(a^{\star_{{\mathcal R}}})$.
Since $W$ is finite and $A_{1}^{\prime}$ and $A_{2}^{\prime}$ are subsets of $W$, therefore let $n_{1}$ and $n_{2}$ be nonnegative integers and $(A_{1,1},s_{1,1}),\ldots,(A_{1,n_{1}},s_{1,n_{1}})$ and $(A_{2,1},s_{2,1}),\ldots,(A_{2,n_{2}},
$\linebreak$
s_{2,n_{2}})$ be couples in $W$ such that $A_{1}^{\prime}=\{(A_{1,1},s_{1,1}),\ldots,(A_{1,n_{1}},s_{1,n_{1}})\}$ and $A_{2}^{\prime}=\{(A_{2,1},s_{2,1}),\ldots,(A_{2,n_{2}},s_{2,n_{2}})\}$.
Since ${\mathcal R}$ is finite, therefore let $a_{1}$ be the least upper bound in ${\mathcal R}$ of the set of all regions $b$ in ${\mathcal R}$ such that $h(b)\subseteq A_{1,1}\cup_{X}\ldots\cup_{X}A_{1,n_{1}}$ and $a_{2}$ be the least upper bound in ${\mathcal R}$ of the set of all regions $b$ in ${\mathcal R}$ such that $h(b)\subseteq A_{2,1}\cup_{X}\ldots\cup_{X}A_{2,n_{2}}$.
Obviously, $h(a_{1})\subseteq A_{1,1}\cup_{X}\ldots\cup_{X}A_{1,n_{1}}$ and $h(a_{2})\subseteq A_{2,1}\cup_{X}\ldots\cup_{X}A_{2,n_{2}}$.
\begin{claim}\label{claim:the:as:are:non:zero}
$a_{1}\not=0_{{\mathcal R}}$ and $a_{2}\not=0_{{\mathcal R}}$.
\end{claim}
\begin{claim}\label{claim:a:is:the:union:of:the:as}
$a=a_{1}\cup_{{\mathcal R}}a_{2}$.
\end{claim}
\begin{claim}\label{claim:the:as:are:dash:with:star:of:a}
$(a_{1},a_{2})\vdash_{{\mathcal R}}a^{\star_{{\mathcal R}}}$.
\end{claim}
By Claims~\ref{claim:the:as:are:non:zero}--\ref{claim:the:as:are:dash:with:star:of:a}, not-$c^{\circ}_{{\mathcal R}}(a)$: a contradiction.
\item In the latter case, let $a_{1},a_{2}$ be regions in ${\mathcal R}$, $a_{1},a_{2}\not=0_{{\mathcal R}}$, such that $a=a_{1}\cup_{{\mathcal R}}a_{2}$ and $(a_{1},a_{2})\vdash_{{\mathcal R}}a^{\star_{{\mathcal R}}}$.
By items~$(1)$ and~$(2)$ of Lemma~\ref{first:lemma:for:the:proposition:relational:representations:of:type:2:of:finite:extended:contact:algebras}, $h^{\prime}(a_{1}),h^{\prime}(a_{2})\not=0_{W}$.
Moreover, by items~$(3)$--$(5)$ of Lemma~\ref{first:lemma:for:the:proposition:relational:representations:of:type:2:of:finite:extended:contact:algebras}, $h^{\prime}(a)=h^{\prime}(a_{1})\cup_{W}h^{\prime}(a_{2})$ and $(h^{\prime}(a_{1}),h^{\prime}(a_{2}))\vdash_{W}h^{\prime}(a)^{\star_{W}}$.
Thus, not-$c^{\circ}_{W}(h^{\prime}(a))$: a contradiction.
\end{itemize}
Hence, $c^{\circ}_{{\mathcal R}}(a)$ iff $c^{\circ}_{W}(h^{\prime}(a))$.
\end{proof}
\\
\\
This completes the proof of Proposition~\ref{proposition:relational:representations:of:type:2:of:finite:extended:contact:algebras}.
\end{proof}
\section{Conclusion}\label{section:Final:remarks}
The above representation theorems for extended contact algebras open new perspectives for region-based theories of space.
We anticipate a number of further investigations.
\\
\\
Firstly, there is the question of the generalization of Proposition~\ref{proposition:relational:representations:of:extended:contact:algebras} to the class of all extended contact algebras, not only the weak algebras.
Can we find necessary and sufficient conditions such that every extended contact algebra that satisfy them can be embedded in the weak extended contact algebra defined over some parametrized frame?
Can these conditions be first-order conditions, or are they intrinsically second-order conditions?
What kind of correspondence can we obtain between topological spaces and parametrized frames within the context of the extended contact relation?
\\
\\
Secondly, there is the question of the generalization of Propositions~\ref{proposition:relational:representations:of:type:1:of:finite:extended:contact:algebras} and~\ref{proposition:relational:representations:of:type:2:of:finite:extended:contact:algebras} to the class of all extended contact algebras, not only the finite algebras.
Can we find necessary and sufficient conditions such that every extended contact algebra that satisfy them can be embedded in the extended contact algebra defined over some equivalence frames of type $1$, or $2$?
Can these conditions be first-order conditions, or are they intrinsically second-order conditions?
\\
\\
Thirdly, following the line of reasoning suggested by Wolter and Zakharyaschev~\cite{Wolter:Zakharyaschev:2000} and furthered in~\cite{Balbiani:Tinchev:Vakarelov:2007,Vakarelov:2007} for what concerns axiomatization/completeness issues and in~\cite{Kontchakov:Nenov:Pratt:Hartmann:Zakharyaschev:2013,Kontchakov:Pratt:Hartmann:Wolter:Zakharyaschev:2010,Kontchakov:Pratt:Hartmann:Zakharyaschev:2010,Kontchakov:Pratt:Hartmann:Zakharyaschev:2014} for what concerns decidability/complexity issues, one may interest in the properties of a quantifier-free first-order language to be interpreted in contact algebras such as the extended contact algebras discussed in this paper.
Can we transfer in this extended setting the axiomatizability results and the decidability results obtained within the more restricted context of the contact relation?
\section*{Acknowledgements}
We make a point of thanking Tinko Tinchev and Dimiter Vakarelov for many stimulating discussions in the field of Contact Logic.
Their comments have greatly helped us to improve the readability of our paper.
Philippe Balbiani and Tatyana Ivanova were partially supported by the programme RILA (contracts 34269VB and DRILA01/2/2015) and the Bulgarian National Science Fund (contract DN02/15/19.12.2016).
\section*{Annex}
{\bf Proof of Claim~\ref{claim:universal:condition:on:sa:and:tb}.}
Let $a^{\prime}\in s_{a}$ and $b^{\prime}\in t_{b}$.
We demonstrate $(a^{\prime},b^{\prime})\not\vdash_{{\mathcal R}}d$.
Suppose $(a^{\prime},b^{\prime})\vdash_{{\mathcal R}}d$.
Since $a^{\prime}\in s_{a}$ and $b^{\prime}\in t_{b}$, therefore $a\leq_{{\mathcal R}}a^{\prime}$ and $b\leq_{{\mathcal R}}b^{\prime}$.
Since $(a^{\prime},b^{\prime})\vdash_{{\mathcal R}}d$, therefore $(a,b)\vdash_{{\mathcal R}}d$: a contradiction.
Hence, $(a^{\prime},b^{\prime})\not\vdash_{{\mathcal R}}d$.
\\
\\
{\bf Proof of Claim~\ref{claim:sa:and:tb:are:filters}.}
By classical results in the theory of filters and ideals~\cite{Givant:Halmos:2009}.
\\
\\
{\bf Proof of Claim~\ref{claim:about:ul:and:vr}.}
We demonstrate $u^{l}$ is an ideal in the Boolean algebra $({\mathcal R},0_{{\mathcal R}},\star_{{\mathcal R}},
$\linebreak$
\cup_{{\mathcal R}})$.
\\
\\
Firstly, suppose $0_{{\mathcal R}}\not\in u^{l}$.
Since $u$ is a filter in the Boolean algebra $({\mathcal R},0_{{\mathcal R}},\star_{{\mathcal R}},
$\linebreak$
\cup_{{\mathcal R}})$, therefore $1_{{\mathcal R}}\in u$.
Since $0_{{\mathcal R}}\not\in u^{l}$, therefore $(1_{{\mathcal R}},0_{{\mathcal R}})\not\vdash_{{\mathcal R}}d$: a contradiction with Proposition~\ref{proposition:consequences:definition:eca}.
Consequently, $0_{{\mathcal R}}\in u^{l}$.
\\
\\
Secondly, suppose $b_{1}^{\prime},b_{2}^{\prime}\in u^{l}$ are such that $b_{1}^{\prime}\cup_{{\mathcal R}}b_{2}^{\prime}\not\in u^{l}$.
Let $a_{1}^{\prime},a_{2}^{\prime}\in u$ be such that $(a_{1}^{\prime},b_{1}^{\prime})\vdash_{{\mathcal R}}d$ and $(a_{2}^{\prime},b_{2}^{\prime})\vdash_{{\mathcal R}}d$.
Since $u$ is a filter in the Boolean algebra $({\mathcal R},0_{{\mathcal R}},\star_{{\mathcal R}},\cup_{{\mathcal R}})$, therefore $a_{1}^{\prime}\cap_{{\mathcal R}}a_{2}^{\prime}\in u$.
Since $b_{1}^{\prime}\cup_{{\mathcal R}}b_{2}^{\prime}\not\in u^{l}$, therefore $(a_{1}^{\prime}\cap_{{\mathcal R}}a_{2}^{\prime},b_{1}^{\prime}\cup_{{\mathcal R}}b_{2}^{\prime})\not\vdash_{{\mathcal R}}d$.
By Proposition~\ref{proposition:consequences:definition:eca}, $(a_{1}^{\prime}\cap_{{\mathcal R}}a_{2}^{\prime},b_{1}^{\prime})\not\vdash_{{\mathcal R}}d$, or $(a_{1}^{\prime}\cap_{{\mathcal R}}a_{2}^{\prime},b_{2}^{\prime})\not\vdash_{{\mathcal R}}d$.
Since $(a_{1}^{\prime},b_{1}^{\prime})\vdash_{{\mathcal R}}d$, $(a_{2}^{\prime},b_{2}^{\prime})\vdash_{{\mathcal R}}d$, $a_{1}^{\prime}\cap_{{\mathcal R}}a_{2}^{\prime}\leq_{{\mathcal R}}a_{1}^{\prime}$ and $a_{1}^{\prime}\cap_{{\mathcal R}}a_{2}^{\prime}\leq_{{\mathcal R}}a_{2}^{\prime}$, therefore by Proposition~\ref{proposition:consequences:definition:eca}, $(a_{1}^{\prime}\cap_{{\mathcal R}}a_{2}^{\prime},b_{1}^{\prime})\vdash_{{\mathcal R}}d$ and $(a_{1}^{\prime}\cap_{{\mathcal R}}a_{2}^{\prime},b_{2}^{\prime})\vdash_{{\mathcal R}}d$: a contradiction.
\\
\\
Thirdly, suppose $b_{1}^{\prime}\in u^{l}$ and $b_{2}^{\prime}\in{\mathcal R}$ are such that $b_{1}^{\prime}\cap_{{\mathcal R}}b_{2}^{\prime}\not\in u^{l}$.
Let $a^{\prime}\in u$ be such that $(a^{\prime},b_{1}^{\prime})\vdash_{{\mathcal R}}d$.
Since $b_{1}^{\prime}\cap_{{\mathcal R}}b_{2}^{\prime}\not\in u^{l}$, therefore $(a^{\prime},b_{1}^{\prime}\cap_{{\mathcal R}}b_{2}^{\prime})\not\vdash_{{\mathcal R}}d$.
Since $(a^{\prime},b_{1}^{\prime})\vdash_{{\mathcal R}}d$ and $b_{1}^{\prime}\cap_{{\mathcal R}}b_{2}^{\prime}\leq_{{\mathcal R}}b_{1}^{\prime}$, therefore by Proposition~\ref{proposition:consequences:definition:eca}, $(a^{\prime},b_{1}^{\prime}\cap_{{\mathcal R}}b_{2}^{\prime})\vdash_{{\mathcal R}}d$: a contradiction.
\\
\\
The proof that $v^{r}$ is an ideal in the Boolean algebra $({\mathcal R},0_{{\mathcal R}},\star_{{\mathcal R}},\cup_{{\mathcal R}})$ is similar.
\\
\\
{\bf Proof of Claim~\ref{claim:about:three:equivalent:conditions}.}
$(1\Rightarrow2)$:
Suppose for all $a^{\prime}\in u$ and for all $b^{\prime}\in v$, $(a^{\prime},b^{\prime})\not\vdash_{{\mathcal R}}d$.
We demonstrate $u^{l}\cap v=\emptyset$.
Suppose $u^{l}\cap v\not=\emptyset$.
Let $b^{\prime\prime}$ be a region in ${\mathcal R}$ such that $b^{\prime\prime}\in u^{l}$ and $b^{\prime\prime}\in v$.
Hence, there exists $a^{\prime}\in u$ such that $(a^{\prime},b^{\prime\prime})\vdash_{{\mathcal R}}d$.
Since $b^{\prime\prime}\in v$, therefore there exists $a^{\prime}\in u$ and there exists $b^{\prime}\in v$ such that $(a^{\prime},b^{\prime})\vdash_{{\mathcal R}}d$: a contradiction.
Thus, $u^{l}\cap v=\emptyset$.
\\
$(2\Rightarrow1)$:
Suppose $u^{l}\cap v=\emptyset$.
We demonstrate for all $a^{\prime}\in u$ and for all $b^{\prime}\in v$, $(a^{\prime},b^{\prime})\not\vdash_{{\mathcal R}}d$.
Suppose there exists $a^{\prime}\in u$ and there exists $b^{\prime}\in v$ such that $(a^{\prime},b^{\prime})\vdash_{{\mathcal R}}d$.
Let $a^{\prime\prime}\in u$ and $b^{\prime\prime}\in v$ be such that $(a^{\prime\prime},b^{\prime\prime})\vdash_{{\mathcal R}}d$.
Hence, there exists $a^{\prime}\in u$ such that $(a^{\prime},b^{\prime\prime})\vdash_{{\mathcal R}}d$.
Thus, $b^{\prime\prime}\in u^{l}$.
Since $b^{\prime\prime}\in v$, therefore $u^{l}\cap v\not=\emptyset$: a contradiction.
Consequently, for all $a^{\prime}\in u$ and for all $b^{\prime}\in v$, $(a^{\prime},b^{\prime})\not\vdash_{{\mathcal R}}d$.
\\
$(1\Rightarrow3)$ and $(3\Rightarrow1)$:
Similar to $(1\Rightarrow2)$ and $(2\Rightarrow1)$.
\\
\\
{\bf Proof of Claim~\ref{claim:the:as:are:non:zero}.}
Suppose $a_{1}=0_{{\mathcal R}}$, or $a_{2}=0_{{\mathcal R}}$.
Without loss of generality, suppose $a_{1}=0_{{\mathcal R}}$.
Since $A_{1}^{\prime}\not=0_{W}$, therefore $n_{1}\geq1$.
Let $i\leq n_{1}$ be a positive integer.
Hence, $A_{1,i}$ is an atom of $(RC(X,\tau),0_{X},\star_{X},\cup_{X})$.
Since $h$ is a surjective embedding of $({\mathcal R},0_{{\mathcal R}},\star_{{\mathcal R}},\cup_{{\mathcal R}},\vdash_{{\mathcal R}})$ in $(RC(X,\tau),0_{X},\star_{X},\cup_{X},\vdash_{X})$, therefore let $b$ be a region in ${\mathcal R}$ such that $h(b)=A_{1,i}$.
Thus, $h(b)\subseteq A_{1,1}\cup_{X}\ldots\cup_{X}A_{1,n_{1}}$.
Consequently, $b\leq_{{\mathcal R}}a_{1}$.
Since $a_{1}=0_{{\mathcal R}}$, therefore $b=0_{{\mathcal R}}$.
Since $A_{1,i}$ is an atom of $(RC(X,\tau),0_{X},\star_{X},\cup_{X})$, therefore $A_{1,i}\not=0_{X}$.
Since $h(b)=A_{1,i}$, therefore $h(b)\not=0_{X}$.
Since $h$ is a surjective embedding of $({\mathcal R},0_{{\mathcal R}},\star_{{\mathcal R}},\cup_{{\mathcal R}},\vdash_{{\mathcal R}})$ in $(RC(X,\tau),0_{X},\star_{X},\cup_{X},\vdash_{X})$, therefore $b\not=0_{{\mathcal R}}$: a contradiction.
Hence, $a_{1}\not=0_{{\mathcal R}}$ and $a_{2}\not=0_{{\mathcal R}}$.
\\
\\
{\bf Proof of Claim~\ref{claim:a:is:the:union:of:the:as}.}
Suppose $a\not=a_{1}\cup_{{\mathcal R}}a_{2}$.
By items~$(1)$ and~$(4)$ of Lemma~\ref{first:lemma:for:the:proposition:relational:representations:of:type:2:of:finite:extended:contact:algebras}, $h^{\prime}(a)\not=h^{\prime}(a_{1})\cup_{W}h^{\prime}(a_{2})$.
Hence, $h^{\prime}(a)\not\subseteq h^{\prime}(a_{1})\cup_{W}h^{\prime}(a_{2})$, or $h^{\prime}(a_{1})\cup_{W}h^{\prime}(a_{2})\not\subseteq h^{\prime}(a)$.
We have to consider two cases.
\begin{itemize}
\item In the former case, since $h^{\prime}(a)=A_{1}^{\prime}\cup_{W}A_{2}^{\prime}$, therefore $A_{1}^{\prime}\cup_{W}A_{2}^{\prime}\not\subseteq h^{\prime}(a_{1})\cup_{W}h^{\prime}(a_{2})$.
Thus, $A_{1}^{\prime}\not\subseteq h^{\prime}(a_{1})\cup_{W}h^{\prime}(a_{2})$, or $A_{2}^{\prime}\not\subseteq h^{\prime}(a_{1})\cup_{W}h^{\prime}(a_{2})$.
Without loss of generality, suppose $A_{1}^{\prime}\not\subseteq h^{\prime}(a_{1})\cup_{W}h^{\prime}(a_{2})$.
Consequently, $A_{1}^{\prime}\not\subseteq h^{\prime}(a_{1})$.
Let $i\leq n_{1}$ be a positive integer such that $(A_{1,i},s_{1,i})\not\in h^{\prime}(a_{1})$.
Hence, $A_{1,i}\not\subseteq h(a_{1})$.
Since $h$ is a surjective embedding of $({\mathcal R},0_{{\mathcal R}},\star_{{\mathcal R}},\cup_{{\mathcal R}},\vdash_{{\mathcal R}})$ in $(RC(X,\tau),0_{X},\star_{X},\cup_{X},\vdash_{X})$, therefore let $b$ be a region in ${\mathcal R}$ such that $h(b)=A_{1,i}$.
Thus, $h(b)\subseteq A_{1,1}\cup_{X}\ldots\cup_{X}A_{1,n_{1}}$.
Consequently, $b\leq_{{\mathcal R}}a_{1}$.
Since $h$ is a surjective embedding of $({\mathcal R},0_{{\mathcal R}},\star_{{\mathcal R}},\cup_{{\mathcal R}},\vdash_{{\mathcal R}})$ in $(RC(X,\tau),0_{X},\star_{X},\cup_{X},\vdash_{X})$, therefore $h(b)\subseteq h(a_{1})$.
Since $A_{1,i}\not\subseteq h(a_{1})$, therefore $h(b)\not=A_{1,i}$: a contradiction.
\item In the latter case, let $(B,t)$ be a couple in $W$ such that $(B,t)\in h^{\prime}(a_{1})\cup_{W}h^{\prime}(a_{2})$ and $(B,t)\not\in h^{\prime}(a)$.
Hence, $(B,t)\in h^{\prime}(a_{1})$, or $(B,t)\in h^{\prime}(a_{2})$.
Without loss of generality, suppose $(B,t)\in h^{\prime}(a_{1})$.
Thus, $B\subseteq h(a_{1})$.
Since $h(a_{1})\subseteq A_{1,1}\cup_{X}\ldots\cup_{X}A_{1,n_{1}}$, therefore $B\subseteq A_{1,1}\cup_{X}\ldots\cup_{X}A_{1,n_{1}}$.
Since $A_{1,1},\ldots,A_{1,n_{1}}$ and $B$ are atoms of $(RC(X,\tau),0_{X},\star_{X},\cup_{X})$, therefore let $i\leq n_{1}$ be a positive integer such that $B=A_{1,i}$.
Consequently, $(B,s_{1,i})\in A_{1}^{\prime}$.
Hence, $(B,s_{1,i})\in A_{1}^{\prime}\cup_{W}A_{2}^{\prime}$.
Since $h^{\prime}(a)=A_{1}^{\prime}\cup_{W}A_{2}^{\prime}$, therefore $(B,s_{1,i})\in h^{\prime}(a)$.
Thus, $B\subseteq h(a)$.
Consequently, $(B,t)\in h^{\prime}(a)$: a contradiction.
\end{itemize}
Hence, $a=a_{1}\cup_{{\mathcal R}}a_{2}$.
\\
\\
{\bf Proof of Claim~\ref{claim:the:as:are:dash:with:star:of:a}.}
Suppose $(a_{1},a_{2})\not\vdash_{{\mathcal R}}a^{\star_{{\mathcal R}}}$.
By item~$(5)$ of Lemma~\ref{first:lemma:for:the:proposition:relational:representations:of:type:2:of:finite:extended:contact:algebras}, $(h^{\prime}(a_{1}),
$\linebreak$
h^{\prime}(a_{2}))\not\vdash_{W}h^{\prime}(a^{\star_{{\mathcal R}}})$.
Hence, $h^{\prime}(a_{1})\cap_{W}h^{\prime}(a_{2})\not\subseteq h^{\prime}(a^{\star_{{\mathcal R}}})$, or there exists a couple $(B,t)$ in $W$ such that $R_{1}(R_{2}((B,t)))$ intersects both $h^{\prime}(a_{1})$ and $h^{\prime}(a_{2})$ and
\linebreak$
R_{1}(R_{2}((B,t)))$ does not intersect $h^{\prime}(a^{\star_{{\mathcal R}}})$.
We have to consider two cases.
\begin{itemize}
\item In the former case, let $(B,t)$ be a couple in $W$ such that $(B,t)\in h^{\prime}(a_{1})$, $(B,t)\in h^{\prime}(a_{2})$ and $(B,t)\not\in h^{\prime}(a^{\star_{{\mathcal R}}})$.
Thus, $B\subseteq h(a_{1})$, $B\subseteq h(a_{2})$ and $B\not\subseteq h(a^{\star_{{\mathcal R}}})$.
Since $h(a_{1})\subseteq A_{1,1}\cup_{X}\ldots\cup_{X}A_{1,n_{1}}$ and $h(a_{2})\subseteq A_{2,1}\cup_{X}\ldots\cup_{X}A_{2,n_{2}}$, therefore $B\subseteq A_{1,1}\cup_{X}\ldots\cup_{X}A_{1,n_{1}}$ and $B\subseteq A_{2,1}\cup_{X}\ldots\cup_{X}A_{2,n_{2}}$.
Since $A_{1,1},\ldots,A_{1,n_{1}}$, $A_{2,1},\ldots,A_{2,n_{2}}$ and $B$ are atoms of $(RC(X,\tau),0_{X},\star_{X},\cup_{X})$, therefore let $i\leq n_{1}$ and $j\leq n_{2}$ be positive integers such that $B=A_{1,i}$ and $B=A_{2,j}$.
Let $u\in Int_{\tau}(B)$.
Consequently, $(A_{1,i},s_{1,i})\in R_{1}(R_{2}((B,u)))$ and $(A_{2,j},s_{2,j})\in R_{1}(R_{2}((B,u)))$.
Hence, $R_{1}(R_{2}((B,u)))\cap A_{1}^{\prime}\not=\emptyset$ and $R_{1}(R_{2}((B,u)))\cap A_{2}^{\prime}\not=\emptyset$.
Since $(A_{1}^{\prime},A_{2}^{\prime})\vdash_{W}h^{\prime}(a)^{\star_{W}}$, therefore $R_{1}(R_{2}((B,
$\linebreak$
u)))\cap h^{\prime}(a)^{\star_{W}}\not=\emptyset$.
Let $(D,v)$ be a couple in $W$ such that $(D,v)\in R_{1}(R_{2}((B,
$\linebreak$
u)))$ and $(D,v)\in h^{\prime}(a)^{\star_{W}}$.
Since $u\in Int_{\tau}(B)$, therefore $B=D$.
Since $h$ is a surjective embedding of $({\mathcal R},0_{{\mathcal R}},\star_{{\mathcal R}},\cup_{{\mathcal R}},\vdash_{{\mathcal R}})$ in $(RC(X,\tau),0_{X},\star_{X},\cup_{X},\vdash_{X})$, therefore $h^{\prime}(a^{\star_{{\mathcal R}}})=h^{\prime}(a)^{\star_{W}}$.
Since $(D,v)\in h^{\prime}(a)^{\star_{W}}$, therefore $(D,v)\in h^{\prime}(a^{\star_{{\mathcal R}}})$.
Thus, $D\subseteq h(a^{\star_{{\mathcal R}}})$.
Since $B\not\subseteq h(a^{\star_{{\mathcal R}}})$, therefore $B\not=D$: a contradiction.
\item In the latter case, let $(B,t)$ be a couple in $W$ such that $R_{1}(R_{2}((B,t)))$ intersects both $h^{\prime}(a_{1})$ and $h^{\prime}(a_{2})$ and $R_{1}(R_{2}((B,t)))$ does not intersect $h^{\prime}(a^{\star_{{\mathcal R}}})$.
Let $(D_{1},u_{1}),(D_{2},u_{2})$ be couples in $W$ such that $(D_{1},u_{1})\in R_{1}(R_{2}((B,t)))$, $(D_{1},u_{1})\in h^{\prime}(a_{1})$, $(D_{2},u_{2})\in R_{1}(R_{2}((B,t)))$ and $(D_{2},
$\linebreak$
u_{2})\in h^{\prime}(a_{2})$.
Consequently, $D_{1}\subseteq h(a_{1})$ and $D_{2}\subseteq h(a_{2})$.
Since $h(a_{1})\subseteq A_{1,1}\cup_{X}\ldots\cup_{X}A_{1,n_{1}}$ and $h(a_{2})\subseteq A_{2,1}\cup_{X}\ldots\cup_{X}A_{2,n_{2}}$, therefore $D_{1}\subseteq A_{1,1}\cup_{X}\ldots\cup_{X}A_{1,n_{1}}$ and $D_{2}\subseteq A_{2,1}\cup_{X}\ldots\cup_{X}A_{2,n_{2}}$.
Since $A_{1,1},\ldots,A_{1,n_{1}}$, $A_{2,1},\ldots,A_{2,n_{2}}$, $D_{1}$ and $D_{2}$ are atoms of $(RC(X,\tau),0_{X},
$\linebreak$
\star_{X},\cup_{X})$, therefore let $i\leq n_{1}$ and $j\leq n_{2}$ be positive integers such that $D_{1}=A_{1,i}$ and $D_{2}=A_{2,j}$.
Hence, $(A_{1,i},s_{1,i})\in R_{1}(R_{2}((B,t)))$ and $(A_{2,j},s_{2,j})\in R_{1}(R_{2}((B,t)))$.
Thus, $R_{1}(R_{2}((B,t)))\cap A_{1}^{\prime}\not=\emptyset$ and $R_{1}(R_{2}((B,t)))\cap A_{2}^{\prime}\not=\emptyset$.
Since $(A_{1}^{\prime},A_{2}^{\prime})\vdash_{W}h^{\prime}(a)^{\star_{W}}$, therefore $R_{1}(R_{2}((B,t)))\cap h^{\prime}(a)^{\star_{W}}\not=\emptyset$.
Since $h$ is a surjective embedding of $({\mathcal R},0_{{\mathcal R}},\star_{{\mathcal R}},\cup_{{\mathcal R}},\vdash_{{\mathcal R}})$ in $(RC(X,\tau),0_{X},\star_{X},\cup_{X},
$\linebreak$
\vdash_{X})$, therefore $h^{\prime}(a^{\star_{{\mathcal R}}})=h^{\prime}(a)^{\star_{W}}$.
Since $R_{1}(R_{2}((B,t)))$ does not intersect $h^{\prime}(a^{\star_{{\mathcal R}}})$, therefore $R_{1}(R_{2}((B,t)))\cap h^{\prime}(a)^{\star_{W}}=\emptyset$: a contradiction.
\end{itemize}
Consequently, $(a_{1},a_{2})\vdash_{{\mathcal R}}a^{\star_{{\mathcal R}}}$.
\end{document}